\documentclass{article}
\usepackage{kr}

\usepackage{times}
\usepackage{soul}
\usepackage{url}
\usepackage{makecell}
\usepackage[hidelinks]{hyperref}
\usepackage[utf8]{inputenc}
\usepackage[small]{caption}
\usepackage{graphicx}
\usepackage{amsmath}
\usepackage{amsthm}
\usepackage{amssymb}
\usepackage{comment}
\usepackage{mathtools}
\usepackage{esvect}
\usepackage{tabularray}
\usepackage{booktabs}
\usepackage{algorithm}
\usepackage{algorithmic}
\usepackage{arydshln}
\usepackage{pifont}
\usepackage{tikz}
\usepackage{ltl}
\usepackage[cmtip,all]{xy}
\usepackage{tikz}
\usetikzlibrary{arrows,calc,decorations.pathmorphing,automata,shapes,positioning,hobby,patterns,patterns.meta}
\usepackage{subcaption}
\usepackage{cleveref}

\usepackage{tabularray}
\usepackage{comment}
\UseTblrLibrary{booktabs}
\usepackage{makecell}
\usepackage{arydshln}
\usepackage{svg}
\usepackage{orcidlink}
\usepackage{marvosym}

\usepackage{todonotes}

\newcounter{sarrow}
\newcommand\xrsquigarrow[1]{%
\stepcounter{sarrow}%
\mathrel{\begin{tikzpicture}[baseline= {( $ (current bounding box.south) + (0,-0.5ex) $ )}]
\node[inner sep=.5ex] (\thesarrow) {$\scriptstyle #1$};
\path[draw,<-,decorate,
  decoration={zigzag,amplitude=0.7pt,segment length=1.2mm,pre=lineto,pre length=4pt}] 
    (\thesarrow.south east) -- (\thesarrow.south west);
\end{tikzpicture}}%
}
\newcommand{\cmark}{\ding{51}}%
\newcommand{\xmark}{\ding{55}}%

\usepackage{xspace}

\newcommand{\ltl}{LTL\xspace}
\newcommand{\kltl}{KLTL\xspace}

\newcommand{\yltlso}{YLTL\textsuperscript{2}\xspace}

\newcommand{\AP}{\text{AP}}
\newcommand{\act}{\mathit{Act}}

\newcommand{\ags}{\text{AG}}
\newcommand{\sov}{\text{SO}}

\newcommand{\true}[0]{\mathit{true}}
\newcommand{\false}[0]{\mathit{false}}

\newcommand{\ldot}{\mathpunct{.}}

\newcommand{\Lang}{\mathcal{L}}

\newcommand{\U}{\LTLuntil}

\newcommand{\K}{\mathtt{K}}

\newcommand{\traces}{\mathit{traces}}

\newcommand{\kripke}{\mathcal{T}}
\newcommand{\ekripke}{\mathcal{E}}

\newcommand{\traceassign}{\Gamma}
\newcommand{\tracevars}{\mathcal{V}}
\renewcommand{\models}{\vDash}
\newcommand{\nmodels}{\nvDash}
\newcommand{\ap}{p}

\newcommand{\donotshow}[1]{}
\newcommand{\tr}{\pi}
\newcommand{\trv}{\tau}
\newcommand{\TV}{\mathit{TV}}
\newcommand{\TR}{\mathit{T}}

\newtheorem{lemma}{Lemma}

\newtheorem{example}{Example}
\newtheorem{theorem}{Theorem}
\newtheorem{proposition}{Proposition}
\newtheorem{definition}{Definition}

\title{An Information-Flow Perspective on Explainability Requirements:\\ Specification and Verification}

\author{%
Bernd Finkbeiner$^1$\and
Hadar Frenkel$^2$\and
Julian Siber$^1$\\
\affiliations
$^1$CISPA Helmholtz Center for Information Security, Saarbrücken, Germany\\
$^2$Bar-Ilan University, Ramat Gan, Israel\\
\emails
finkbeiner@cispa.de,
hadar.frenkel@biu.ac.il,
julian.siber@cispa.de
}

\begin{document}

\maketitle

\begin{abstract}
Explainable systems expose information about why certain observed effects are happening to the agents interacting with them. We argue that this constitutes a positive flow of information that needs to be specified, verified, and balanced against negative information flow that may, e.g., violate privacy guarantees. Since both explainability and privacy require reasoning about knowledge, we tackle these tasks with epistemic temporal logic extended with quantification over counterfactual causes. This allows us to specify that a multi-agent system exposes enough information such that agents acquire knowledge on why some effect occurred. We show how this principle can be used to specify explainability as a system-level requirement and provide an algorithm for checking finite-state models against such specifications. We present a prototype implementation of the algorithm and evaluate it on several benchmarks, illustrating how our approach distinguishes between explainable and unexplainable systems, and how it allows to pose additional privacy requirements. 
\end{abstract}

\section{Introduction}\label{sec:intro}
Contemporary autonomous systems are increasingly complex and opaque, yet also deployed in consequential applications such as hiring~\cite{Amazon}, healthcare~\cite{Davenport}, and criminal sentencing~\cite{ProPublica16}. This tension has led to an extensive inquiry into methods that provide explanations for the behavior of these systems, such that agents interacting with, e.g., a hiring system may know \emph{why} their application is rejected~\cite{MershaLWAK24}. 

Although an explanation may provide critical and actionable recourse to one agent, it may also reveal private information about another~\cite{NguyenHRNNYN25}. For instance, a rejected job application may be explained by the applicant's pay requirement being over a certain threshold, but if the applicant observes a sufficiently similar application being accepted, they can infer the future salary of the other agent. This line of reasoning leads to an inherent tradeoff between explainability and privacy, which we study in this paper.

Privacy requirements are commonly formalized as \emph{information-flow policies}~\cite{GoguenM82a,KozyriCM22}, i.e., a system must not only restrict direct access to private information but also restrict propagation of this information through indirect channels. We approach the formalization of explainability requirements from the same perspective and express them as end-to-end information-flow policies in a specification language for multi-agent systems. On a high level, such an explainability requirement defines the necessary flow of information in the following way: It specifies what needs to be explained (the \emph{explanandum}), how it needs to be explained (the \emph{explanans}), and when it needs to be explained, to whom.
We propose a verification algorithm that can then check whether a given system ensures sufficient flow of information to meet a requirement in our specification language. With the same language, we can also express privacy requirements, such that we can similarly check that explainability does not come at the expense of security.

Our specification language extends epistemic temporal logic~\cite{FaginHMV1995}. Epistemic and temporal logics have been popular frameworks for studying information-flow security~\cite{HalpernO08,BalliuDG11,ClarksonFKMRS14,CoenenFHH19}. This inspires us to study the flow of information in explainable systems through the same lens. We build on the popular framework of counterfactual explanations~\cite{HalpernP05b,Wachter18} and instrument our logic with operators to reason about temporal causes~\cite{FinkbeinerFMS24}. These are based on a temporal variant of actual causation~\cite{Halpern16}, which uses counterfactual reasoning~\cite{Lewis73} to explain a given  execution and has received significant attention in the literature on explainable AI~\cite{Miller19}. A temporal cause can for instance be described symbolically by the formula $\LTLprevious a_1 \land \LTLpastfinally a_2$, which means that action $a_1$ at the previous time point and action $a_2$ at any earlier time point have jointly caused some effect: The cause needs to be satisfied and describe the minimal changes necessary to obtain a counterfactual  execution where the effect does not happen.

The combination of counterfactual, epistemic and temporal reasoning allows us to express explainability requirements of the following form:
\begin{align*}
    \LTLglobally \big(\psi \rightarrow \exists X .\, \K_a (X \xrsquigarrow{\mathit{Act}(a)} \psi)\big) \enspace ,
\end{align*}
which states that the explanandum $\psi$ is explainable to agent $a$ whenever $\psi$ occurs, via the temporal cause $X$ serving as explanans. The cause $X$ is constrained to reason only over the actions $\mathit{Act}(a)$ of agent~$a$, which is why we term this requirement \emph{Internal Causal Explainability (ICE)}. The requirement uses the temporal operator $\LTLglobally$ to enforce its constraint on every time point of an execution, and the epistemic operator $\K_a$ to express its key epistemic component: The semantics of this latter operator require that the same property $X$ causes $\psi$ on all  executions that are indistinguishable to agent $a$, which means that agent $a$ has acquired knowledge of the associated causal dependency. 

\paragraph{Contributions and Outline.} We give a detailed example to illustrate how ICE encodes explainability in an information-flow sense in Section~\ref{sec:motivation}, after establishing necessary preliminaries in Section~\ref{sec:prelims}. The motivating example highlights how explainability requirements, such as ICE, require tradeoffs when juxtaposed with privacy requirements placed on the same system. We then introduce the formal details of our logic in Section~\ref{sec:main}, where we discuss other explainability requirements beside ICE, and outline an algorithm to verify whether a given system satisfies a requirement specified in our logic. The main challenge for the algorithm is the second-order quantification that ranges over sets of traces, as related logics with unrestricted quantifiers of this kind cannot be verified automatically~\cite{BeutnerFFM23}. We show how to exploit the fact that causes are uniquely determined to encode the second-order quantifiers over sets of traces into the decidable satisfiability problem of a temporal logic with first-order quantification over atomic propositions only.
We then report on experiments with a prototype implementation of our approach in Section~\ref{sec:experiments}. Our prototype can verify both explainability and privacy requirements as introduced throughout the paper, on multi-agent systems with up to several thousand states. We use classic games and an auction system for these experiments.
Last, we discuss related work in Section~\ref{sec:related} and close with a short summary and outlook on future work in Section~\ref{sec:summary}.

\section{Preliminaries}\label{sec:prelims}

We recall the formal background on transition systems as models of multi-agent systems, temporal logics for specifying system requirements, and temporal causality for defining counterfactual dependencies between temporal properties.

\subsubsection{Multi-Agent Systems.}

We consider \emph{transition systems} as the fundamental model of the logics we will study in this paper. A transition system is a tuple $\kripke = (S,S_0,\Delta,\AP,\Lambda)$, where $S$ is a finite set of \emph{states}, $S_0$ is a set of \emph{initial states}, $\Delta: S \mapsto 2^S$ is a \emph{transition function} such that $\Delta(s) \neq \emptyset$ for all states $s \in S$, $\AP$ is a set of \emph{atomic propositions}, and $\Lambda : S \times S \mapsto 2^\AP$ is a \emph{labeling function} marking edges with atomic propositions. Executions of a system are modeled as follows. A \emph{path} $\rho = \rho[0] \rho[1] \ldots \in S^\omega$ of $\kripke$ is an infinite sequence of states following the transition function: $\rho[i+1] \in \Delta(\rho[i])$ for all \emph{time points} $i \in \mathbb{N}$. The \emph{trace} $\pi =  \pi[0] \pi[1] \ldots \in (2^\AP)^\omega$ of a path $\rho$ is the sequence of corresponding labels, i.e., we have $\pi[i] = \Lambda(\rho[i],\rho[i+1])$ for all $i \in \mathbb{N}$.  Let $\Pi(\kripke)$ denote the set of traces of initial paths, i.e., of $\rho$ such that $\rho[0] \in S_0$. For some trace $\pi$, $\pi[0,n] \in S^*$ is its \emph{prefix} of length $n+1$. For two traces $\pi,\pi \in S^\omega$ and $A \subseteq \AP$ we write $\pi =_A \pi'$ if $\pi[i]\cap A = \pi'[i]\cap A$ for all time points $i \in \mathbb{N}$.
We model a multi-agent system as an \emph{extended transition system} $\ekripke = (\kripke,\Omega,\act)$ that includes an \emph{observation map} $\Omega : \ags \mapsto 2^\AP$ to reason about the observations of a set of agents $\ags$ and an \emph{action map} $\act : \ags \mapsto 2^\AP$ to reason about their controllable actions. We will use the shorthand $\act$ for the image of $\ags$ under $\act$, i.e., the set of all actions. We call every atomic proposition that is not an action a system \emph{output}. For some agent $a \in \ags$, $\Omega(a)$ describes the set of atomic propositions that are observable to agent $a$, and $\act(a)$ describes which actions are controllable by $a$. We assume $\act(a) \subseteq \Omega(a)$. For some trace $\pi$ of $\kripke$, $\Omega_a(\pi) \in (2^\AP)^\omega$ are the partial observations of $a$ along the trace: $\Omega_a(\pi)[i] = \pi[i] \cap \Omega(a)$. We further assume that all actions are possible from every state (although they may have no effect): For all $s \in S$ and $A \subseteq \act$, there is an $s' \in S$ such that $s' \in \Delta(s)$ and $A = \Lambda(s,s') \cap \act$. We say that a multi-agent system is \emph{deterministic} if there exists exactly one such $s'$ for all $s \in S$ and $A \subseteq \act$. The set of traces of $\ekripke = (\kripke,\Omega)$ is denoted $\Pi(\ekripke) = \Pi(\kripke)$.

\begin{example}
    Consider the multi-agent system $\mathcal{A}_\mathit{explain}$ shown in Figure~\ref{fig:dutch}. It models an explainable auction between three bidders. Nodes and edges depict states and the transition function, respectively. The single initial state is indicated through an incoming edge without a source state. For brevity, edge-labels use symbolic notation for actions (left of the bar) and explicit sets for the outputs (right of the bar). If there are no outputs, we only depict the action constraints. A trace of this system is $\pi = \{o,b_1,b_2,e\}(\{o,b_1,b_2\})^\omega$, which corresponds to path $\rho_1 = \text{\emph{init}}\,(\text{\emph{win}}_1)^\omega$ or $\rho_2 = \text{\emph{init}}\,(\text{\emph{win}}_2)^\omega$, since the set $\{o,b_1,b_2\}$ satisfies $o \land \mathit{b}_1$ as well as $o \land \mathit{b}_2$. The $\omega$-superscript denotes infinite repetition of the subsequence.
\end{example}

\subsubsection{Temporal Logics.}
The basis of our logic is the epistemic temporal logic \kltl, which extends \emph{Linear Temporal Logic (\ltl)} with a knowledge modality. We defer KLTL and start with a definition of \ltl~\cite{Pnueli77}. We include past operators, which do not increase expressive power~\cite{LichtensteinPZ85}, but will make some formulas more readable. The syntax is given as follows:
\begin{align}\varphi \Coloneqq p \mid \neg \varphi \mid \varphi \lor \varphi \mid \LTLnext \varphi \mid \varphi \U \varphi \mid \LTLprevious \varphi \mid \varphi \U^- \varphi \enspace,\label{gram:ltl}\end{align}
where $p \in \AP$ is an atomic proposition. Additionally, \ltl includes the following derived operators: Boolean constants ($\true$, $\false$) and connectives ($\lor$, $\rightarrow$, $\leftrightarrow$), the temporal operator `Eventually' ($\LTLeventually \varphi \equiv \true \LTLuntil \varphi$) as well as its dual, `Globally' ($\LTLglobally \varphi \equiv \lnot \LTLeventually \lnot \varphi $), and the derived past operators `Once' ($\LTLpastfinally \varphi \equiv \true \LTLuntil^- \varphi$) and `Historically' ($\LTLpastglobally \varphi \equiv \lnot \LTLpastfinally \lnot \varphi $). The semantics of an \ltl formula $\varphi$ is defined with respect to a trace $\pi$ and a time point $i$  as follows:
\begin{alignat*}{2}
	&\pi,i \models p	&~\text{ iff }~	&p \in \pi[i],\\
	&\pi,i \models \lnot \varphi	&~\text{ iff }~	& \pi,i \nmodels \varphi,\\
	&\pi,i \models \varphi_1 \lor \varphi_2	&~\text{ iff }~	& \pi,i \models \varphi_1 \text{ or } \pi,i \models \varphi_2,\\
	&\pi,i \models \LTLnext \varphi	&~\text{ iff }~	& \pi,i+1 \models \varphi,\\
	&\pi,i \models \LTLprevious \varphi	&~\text{ iff }~	& \pi,i-1 \models \varphi \text{ and } i > 0,\\
	&\pi,i \models \varphi_1 \LTLuntil \varphi_2	&~\text{ iff }~	& \exists k \geq i: \pi,k \models \varphi_2 \text{ and} \\  & & &\forall i \leq j < k: \pi,j \models \varphi_1,\\
	&\pi,i \models \varphi_1 \LTLuntil^- \varphi_2	&~\text{ iff }~	& \exists g \leq i: \pi,g \models \varphi_2 \text{ and} \\  & & &\forall g \leq h < i: \pi,h \models \varphi_1  .
\end{alignat*}
We call the combination of trace and time point an \emph{anchor point}. System-level satisfaction is based on a universal application of the trace semantics: $\kripke$ \emph{satisfies} $\varphi$, denoted by $\kripke \models \varphi$, iff for all traces $\pi \in \Pi(\kripke): \pi,0 \models \varphi$. The \emph{language} $\Lang(\varphi)$ is the set of all traces satisfying the \ltl formula $\varphi$.

The \emph{epistemic temporal logic \kltl}~\cite{FaginHMV1995} extends \ltl with a knowledge modality that expresses the knowledge of an agent $a \in \ags$, i.e., it adds the rule $\K_a \, \varphi$ to Grammar~\ref{gram:ltl}. The semantics of a \kltl formula is as for LTL, but additionally refers to an extended transition system $\ekripke = (\kripke,\Omega)$. For the epistemic operator $\K_a$, we have:
\begin{alignat*}{2}
	&\ekripke,\pi,i \models \K_a \, \varphi	&~\text{ iff }~	& \forall \pi' \in \Pi(\kripke): (\Omega_a(\pi)[0,i] =\\ & & & \Omega_a(\pi')[0,i]) \rightarrow \ekripke,\pi',i \models \varphi \enspace .
\end{alignat*}
Hence, agent $a$ has knowledge of some property $\varphi$ on a trace $\pi$ at point $i$, expressed through the formula $\K_a \, \varphi$, if this property holds on all traces that are indistinguishable for $a$ from $\pi$ up to this point. 
This corresponds to the so-called synchronous perfect recall semantics~\cite{MeydenS99,HalpernMV04}, since the agents can distinguish prefixes of different length and based on divergence at any point in the past.  Similar to LTL, we have $\ekripke \models \varphi$ iff for all traces $\pi \in \Pi(\kripke): \ekripke,\pi,0\! \models \varphi$.

\begin{figure}[t]
		\centering
  \begin{tikzpicture}[->,shorten >=1pt,semithick,auto,node distance=3.5cm, on grid,initial text=,
			every state/.style={minimum size=27pt,inner sep=1pt}]
         \node[state](1){init};
         \node[state,left = of 1](2){win$_1$};
         \node[state,below = of 1,yshift=2em](3){win$_2$};
         \node[state,right = of 1](4){win$_3$};
         \node[above left = 1 and 0.7 of 1](i){};
         \path[->,draw](1) edge[loop above,looseness=12] node{$\lnot o \lor (\lnot\mathit{b}_1 \land \lnot\mathit{b}_2 \land \lnot\mathit{b}_3)$} (1)
         (1) edge[bend right=9] node[swap,pos=0.5]{$o \land \mathit{b}_1 \mid \boldsymbol{E}$} (2)
         (2) edge[bend right=9] node[swap,pos=0.5]{$\lnot o \, | \, \{\mathit{w}_1\}$} (1)
         (2) edge[loop below,looseness=12] node{$o$} (2)
         (1) edge[bend left=9] node[pos=0.6]{$o \land \mathit{b}_2 \mid \boldsymbol{E}$} (3)
         (3) edge[bend left=9] node[pos=0.4]{$\lnot o \, | \, \{\mathit{w}_2\}$} (1)
         (1) edge[bend left=9] node[pos=0.5]{$o \land \mathit{b}_3 \mid \boldsymbol{E}$} (4)
         (4) edge[bend left=9] node[pos=0.5]{$\lnot o \, | \, \{\mathit{w}_3\}$} (1)
         (3) edge[loop left,looseness=12] node{$o$} (3)
         (4) edge[loop below,looseness=12] node{$o$} (4);
         \path[->,draw](i) edge[]  (1);
         \end{tikzpicture}
         \caption{A multi-agent system modeling a Dutch auction with three agents bidding through actions $b_1,b_2,b_3$ and one auctioneer opening and closing the auction with action $o$. The winner is announced through outputs $w_1,w_2,w_3$. We consider three versions of the auction in this paper. In the first version $\mathcal{A}_\mathit{blind}$ an agent $i$ observes only their own action, output, and the auctioneer: $b_i,w_i,o$. In the second version $\mathcal{A}_\mathit{public}$ an agent $i$ additionally observes all actions $\act$. For both of these versions, the explanatory output $\boldsymbol{E}$ is empty, while we have $\boldsymbol{E} = \{e\}$ for the third version $\mathcal{A}_\mathit{explain}$, which otherwise is like $\mathcal{A}_\mathit{blind}$ except $e$ is observable by all agents.}
    \label{fig:dutch}
\end{figure}
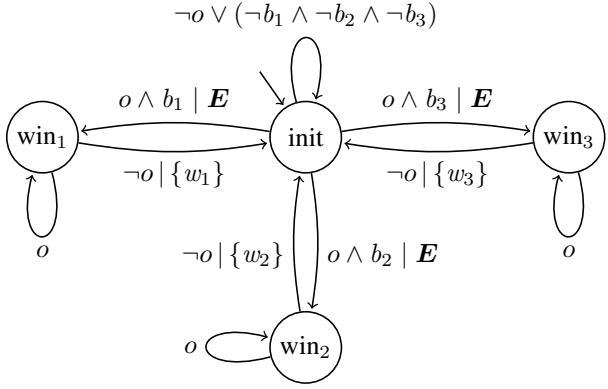

\subsubsection{Temporal Causality.} Since multi-agent systems are sequential processes where the timing of actions is causally relevant, we adapt temporal causality~\cite{FinkbeinerFMS24} to express causal dependencies between temporal properties. This utilizes counterfactual reasoning based on similarity~\cite{Lewis73}. Similarity is defined with a trace-wide extension of the \emph{symmetric difference} $A \oplus B = (A \setminus B) \cup (B \setminus A)$ of two sets $A,B$. For two traces $\pi,\pi' \in (2^\AP)^\omega$ and some $A \subseteq \AP$, we define $\pi \oplus_A \pi' = \{(a,i) \in A \times \mathbb{N} \mid a \in \pi[i] \oplus \pi'[i] \}$. For three traces $\pi,\pi',\pi'' \in (2^\AP)^\omega$ and $A \subseteq \act$ we say $\pi$ is at least as similar to $\pi'$ as $\pi''$ over $A$, denoted with $\pi \leq_{\pi'}^{A} \pi''$, iff $(\pi \oplus_A \pi') \subseteq (\pi'' \oplus_A \pi')$.

We now have everything we need to define temporal causality: The \emph{cause} for some temporal property $\psi$ at time point $i$ on a trace $\pi$ of system $\ekripke = (\kripke,\Omega)$ with respect to $A \subseteq \act$ is the following temporal property:
\begin{align*}
    \mathit{Cause}(\psi,\pi,i,A) = \{\,&\pi'  \in (2^{A})^\omega \mid \forall \pi'' \in \Pi(\kripke).\\ &(\pi'' \leq_\pi^{A} \pi' \land \pi =_{(\act\setminus A)} \pi'')\\ &\rightarrow \ekripke,\pi'',i \models \psi \,\}\enspace.
\end{align*}

A temporal cause describes the largest set of action sequences satisfying the effect $\psi$ that is downward closed in $(\Pi(\kripke),\leq_\pi^A)$, i.e., all sequences with traces that satisfy the effect $\psi$ such that all at least as similar traces also satisfy $\psi$. Causes capture semantically what actions of a trace need to be changed to negate the effect. The parameter $A$ allows to constrain which actions can be changed and which actions are fixed between the traces. A cause may be described symbolically as the language of a temporal logic formula.

\begin{example}\label{ex:prelim}
Consider again system $\mathcal{A}_\mathit{explain}$ in Figure~\ref{fig:dutch}, effect $\psi = \LTLeventually w_1$ and trace $\pi = \{o,b_1,e\}\{o\}\{o,b_1\}\{w_1\}\{\}^\omega$. We have that $\mathit{Cause}(\psi,\pi,0,\{b_1\})$ is the language $\Lang(b_1 \lor \LTLnext \LTLnext b_1)$. This is because all traces that have action $b_1$ at either the first or third point and are equal to $\pi$ on all other actions satisfy the effect $\psi$. Other traces that satisfy $\psi$ such as $\pi' = \{o\}\{o,b_1,e\}\{o\}\{w_1\}\{\}^\omega$ are less similar to $\pi$ with respect to $\{b_1\}$ than $\{o\}\{o\}\{o\}\{\}^\omega$, which does not satisfy $\psi$, and hence these are not included in $\mathit{Cause}(\psi,\pi,0,\{b_1\})$.
\end{example}

\section{From Explanations to Explainability}\label{sec:motivation}

In this section, we give a high-level overview of this work. We particularly focus on delimiting the concepts of (individual) \emph{explanations}, the information flow-based system requirement we call \emph{explainability}, and how our work allows to identify tradeoffs between explainability and privacy.

We illustrate these concepts with the multi-agent systems modeling several versions of a \emph{Dutch auction} depicted in Figure~\ref{fig:dutch}. In a Dutch auction, there are a number of bidders that compete for a resource, and an auctioneer that opens and closes the auction. The auctioneer sets an initial price and decrements the price in every time step until reaching a lower limit price. Participants can place their bids at any time point, with the first bidder in a given auction cycle winning the resource. The systems are nondeterministic for the case that multiple agents place the first bid in a given round. The different versions of the auction differ with respect to the explainability and privacy guarantees that they provide, as we outline in the following.

\subsection{Explanations}\label{subsec:explanations}
In case a bidder does not win an auction, they may be interested in an explanation for this outcome. We consider explanations that answer counterfactual queries, e.g.: What did the agent need to do differently to obtain the desired outcome? The answer to such a query may include the temporal behavior of the agent, e.g., bidding \emph{earlier}. Hence, we use temporal causes (cf.\ Section~\ref{sec:prelims}) as the semantic content of an answer, i.e., as explanantia. Consider a trace $\pi$ of the system $\mathcal{A}_\mathit{blind}$ where Bidder~1 loses the first cycle to Bidder~2.

\begin{center}
\setlength{\tabcolsep}{5.5pt}
\begin{tabular}{c c | c | c | c | c}
 Auctioneer: & $\{o\}$ &  $\{o\}$ &  $\{o\}$ &   $\{\}$ & $\{\}^\omega$ \\ [0.5ex] 
 Bidder 1: & $\{\}$ &  $\{\}$ &   $\{b_1\}$ & $\{\}$ & $\{\}^\omega$ \\  [0.5ex] 
 Bidder 2: & $\{\}$ &  $\{b_2\}$ &  $\{\}$ &  $\{w_2\}$ & $\{\}^\omega$ \\  [0.5ex] 
 Bidder 3: & $\{\}$ &  $\{\}$ &  $\{b_3\}$ &  $\{\}$ & $\!\{\}^\omega$
\end{tabular}
\end{center}
The auction cycle is open while the auctioneer chooses the opening action $o$, in this way modeling them setting the initial and limit prices. During these first three time points, Bidder 2 bids at the second time point with action $b_2$ and the remaining agents bid one step later with actions $b_1$ and $b_3$.

If we restrict counterfactual actions to range only over the actions of Bidder~1 themselves, it is easy to see that bidding at either the first or second time point would have allowed them to win the auction cycle, although the latter only through an advantageous resolution of nondeterminism on the action sequence. We express the counterfactual dependency between $\lnot w_1$ and its cause as follows: 
$$\varphi_{\mathit{cf}} = (\LTLprevious^2 ( \lnot b_1 \land \LTLprevious \lnot b_1))\, \smash{\xrsquigarrow{\{b_1\}}}\, \lnot w_1 \enspace .$$
The formula states that the minimal changes to the  execution to flip the truth of outcome $\lnot w_1$ at the fourth time point with all actions fixed but $\{b_1\}$ also flip the truth of $\LTLprevious^2 ( \lnot b_1 \land \LTLprevious \lnot b_1)$, i.e., Bidder~1 not bidding at two and three time points before. We use an $i$ superscript to shorten sequences of length $i$ of the same operator.

The goal of an explanation for $\lnot w_1$ is to change the epistemic state of Bidder 1 from not knowing $\varphi_{\mathit{cf}}$ to knowing $\varphi_{\mathit{cf}}$, i.e., we want to ensure that $\K_{\text{Bidder\,1}}(\varphi_{\mathit{cf}})$ holds. This knowledge depends on the observations Bidder 1 can make during the execution of $\pi$: They have knowledge only of formulas that hold on all indistinguishable executions at this time point. In the blind auction $\mathcal{A}_\mathit{blind}$ where Bidder 1 can only see their own bids and the auctioneer, the execution of $\pi$ is indistinguishable from the execution of $\pi'$ where Bidder~2 bids with the following sequence.

\begin{center}
\setlength{\tabcolsep}{7.5pt}
\begin{tabular}{c c | c | c | c | c }
 Bidder 2: & $\{b_2\}$ &  $\{\}$ & $\{\}$ &  $\{w_2\}$ &  $\{\}^\omega$
\end{tabular}
\end{center}
On this  execution, the cause for $\lnot w_1$ at the fourth time point is reduced to $\LTLprevious^3 \lnot b_1$. Hence, Bidder 1 does not know the cause for $\lnot w_1$ in $\mathcal{A}_\mathit{blind}$ at this point, i.e., $\K_{\text{Bidder\,1}}(\varphi_{\mathit{cf}})$ does not hold. What can the auction system do such that Bidder~1 gains knowledge of the cause for $\lnot w_1$? The key lies in providing additional observations to the agent that serve as \emph{explanations}. One possible set of explanatory observations for $\lnot w_1$ is $\{b_2,b_3\}$, i.e., the bidding actions of all other agents. In system $\mathcal{A}_\mathit{public}$ where they are observable, Bidder~1 can distinguish $\pi$ from all other traces and hence $\K_{\text{Bidder\,1}}(\varphi_{\mathit{cf}})$ holds on $\pi$ at the fourth time point.

\subsection{Explainability}\label{subsec:explainability}

The previous section has outlined how explanations can transport knowledge about counterfactual causes at a specific time point in a given  execution. To go from explanations to explainable systems we need to define under which circumstances this knowledge should be available to which agents. In the Dutch auction system, we may for instance require that Bidder~1 knows the cause for a loss whenever the auction closes. We can again use temporal operators to express these timing requirements, as we already did in the previous section to describe that the temporal behavior of Bidder~1 is a causal antecedent for their loss at the fourth time point. A complicating factor is that a given effect can have an arbitrary number of causes at different time points. For instance, shifting $\pi$ by duplicating its first time point includes an additional action in the cause for the loss of Bidder~1, which can be spun arbitrarily further. Hence, it does not suffice to use a finite number of explicit causal antecedents in a system-wide explainability requirement. We solve this by extending the logic with second-order quantifiers that allow to quantify over sets of traces. Such a set can be instantiated with different causes at different time points. Our considerations on timing and cause quantification result in the following requirement for the Dutch auction systems:
$$ \LTLglobally \big((\lnot w_1 \land \lnot o \land \LTLprevious o ) \rightarrow \exists X .\, \K_{\text{Bidder\,1}}(X \xrsquigarrow{\{b_1\}} \lnot w_1)\big) \enspace . $$
The requirement states that at all future time points (enforced through the temporal operator $\LTLglobally\,$), whenever the auction was open in the previous time point, is now closed, and Bidder~1 has not won the cycle, then there is a property $X$ such that Bidder~1 knows that $X$ is the cause for their loss, i.e., $X$ is the cause on all traces that are indistinguishable for Bidder~1. 
The cause $X$ is constrained to action $b_1$ of Bidder~1. It is an instance of ICE as introduced in Section~\ref{sec:intro}.

It is easy to see that the auction system $\mathcal{A}_\mathit{blind}$ does not satisfy ICE, since the fourth time point on  execution $\pi$ as discussed in Section~\ref{subsec:explanations} is a counterexample to the system-wide requirement. In $\mathcal{A}_\mathit{public}$, Bidder~1 observes everything about an execution except how nondeterminism is resolved when multiple agents bid first simultaneously. In these instances, the empty set is a valid causal antecedent on all indistinguishable traces and hence $X$ can be instantiated with it. This effectively means that Bidder~1 either knows which actions would have resulted in them winning the auction, or knows that their actions were already optimal and their loss is attributable to unmodeled actions such as randomness. Hence, we have that $\mathcal{A}_\mathit{public}$ satisfies ICE.

\subsection{Balancing Explainability and Privacy}\label{subsec:privacy}
Although $\mathcal{A}_\mathit{public}$ satisfies the explainability requirement, this comes at the cost of exposing all actions of the other agents. This is clearly unsatisfactory, as systems may place privacy requirements alongside explainability. For instance, we may require that Bidder~1 never knows whether Bidder~2 places a bid at a given time point. We can express this system-wide  requirement utilizing as $ \LTLglobally \big(\lnot \K_{\text{Bidder\,1}}(b_2)\big)$,
which formally requires that Bidder~1 should never be able to distinguish a given  execution from another where the value of $b_2$ is flipped. We call this parametric privacy notion $b_2$\emph{-privacy}.
In terms of privacy, the desirability of $\mathcal{A}_\mathit{public}$ and $\mathcal{A}_\mathit{blind}$ is now exactly opposite as for explainability: $\mathcal{A}_\mathit{public}$ clearly does not satisfy it as Bidder~1 can directly observe $b_2$. $\mathcal{A}_\mathit{blind}$ does satisfy it because Bidder~1 cannot distinguish between the actions of Bidders~2 and~3. 

It turns out that making the set $\{b_2,b_3\}$ observable was too impetuous and a smaller set would have been sufficient without sacrificing privacy: System $\mathcal{A}_\mathit{explain}$ adds the output $e$, which gets broadcast to all agents whenever the first bid of an open auction comes in. On the one hand, this still provides the information flow required to identify the cause for a loss of Bidder~1, such that ICE is satisfied by the system. On the other hand, this hides who placed the highest bid, such that the system also satisfies  $b_2$-privacy. 

\section{A Logic for Explainability Requirements}\label{sec:main}

In this section, we dig deeper into the formal details of our logic for expressing explainability requirements, which we call \yltlso.  We start with outlining its syntax and semantics in Section~\ref{sec:synsem}. We then illustrate how the logic can be used to specify a number of different explainability requirements in Section~\ref{sec:reqs}. Last, we describe a model-checking algorithm for the logic in Section~\ref{sec:mc}.

\subsection{Syntax and Semantics of \yltlso}\label{sec:synsem}

\yltlso is an extension of \kltl, i.e., linear temporal logic with the knowledge modality $\K_a$. Hence, the syntax and semantics of all shared operators are as described in Section~\ref{sec:prelims}.

\paragraph{Syntax.} \yltlso extends \kltl with second-order quantification over sets of traces and allows these sets to be used in \emph{causal predicates} \smash{$X$ \raisebox{-1pt}{$\xrsquigarrow{A}$} $\varphi$}, which require $X$ to be the cause for $\varphi$ over a set of actions $A$ at the current time point of a given trace. We assume a set of second-order variables $\sov$ be given with the set of atomic propositions $\AP$ and agents $\ags$. The syntax of \yltlso is as follows:
\begin{align*}
	\varphi \Coloneqq \; &p  \mid \neg \varphi \mid \varphi \land \varphi \mid \LTLnext \varphi \mid \varphi \LTLuntil \varphi \mid \LTLprevious \varphi \mid \varphi \U^- \varphi \mid\\ &\K_a \, \varphi \mid \exists X . \, \varphi \mid X \xrsquigarrow{A} \varphi \enspace ,
\end{align*}
where $p \in \AP$ is an atomic proposition, $a \in \ags$ is an agent, $X \in \sov$ is a second-order variable, and $A \subseteq \AP$ is a subset of atomic propositions. \yltlso includes the same derived operators as \ltl (cf. Section~\ref{sec:prelims}), as well as universal second-order quantification derived as $\forall X. \, \varphi \equiv \lnot \exists X. \, \lnot \varphi $. We say a \yltlso formula is \emph{well-formed} iff all of its subformulas \smash{$X$ \raisebox{-1pt}{$\xrsquigarrow{A}$} $\varphi$}
are in the scope of an existential quantifier $\exists X$, that is, all second-order variables are bound to a quantifier. 

\paragraph{Semantics.} The semantics of all operators shared between \yltlso and \kltl are as defined in Section~\ref{sec:prelims}, but they additionally refer to a second-order assignment $\Theta : \sov \mapsto \mathbb{P}(\Sigma^\omega)$, which is a mapping from second-order variables to sets of traces. This assignment plays a central part in the semantics of the new operators in \yltlso:
\begin{alignat*}{2}
	&\ekripke,\pi,i,\Theta \models \exists X .	\,\varphi &~\text{ iff }~	&\exists Y \subseteq (2^\AP)^\omega\text{ s.t. }\\& & &\ekripke,\pi,i,\Theta[X \mapsto Y] \models \varphi,\\
	&\ekripke,\pi,i,\Theta \models X \xrsquigarrow{A} \varphi	&~\text{ iff }~	&
\Theta(X) = \mathit{Cause}(\varphi,\pi,i,A).
\end{alignat*}
Hence, 
the semantics of 
the second-order quantification $\exists X .	\,\varphi$ requires that there exists a set of traces, assigned to the second-order variable $X$, such that the subformula $\varphi$ is satisfied. In the subformula, this second-order variable may be used in an arbitrary number of causal predicates \smash{$X$ \raisebox{-1pt}{$\xrsquigarrow{A}$} $\varphi$} and it has to qualify as a cause in all of them. This essentially allows to succinctly specify equality of causes at an arbitrary (even infinite) number of anchor points in the scope of a single quantifier. For an extended transition system $\ekripke$ and a well-formed \yltlso formula $\varphi$, we have $\ekripke \models \varphi$ iff for all traces $\pi \in \Pi(\kripke): \ekripke,\pi,0,\mathfrak{E} \models \varphi$, where $\mathfrak{E}$ denotes an empty second-order assignment that maps all variables to $\bot$. The \yltlso \emph{model-checking} problem is to decide whether $\ekripke \models \varphi$ for such an $\ekripke$ and $\varphi$.

\subsection{Formalizing Explainability Requirements}\label{sec:reqs}

Besides ICE, \yltlso can be used to define other explainability requirements. We establish some results on entailment between these requirements and other notions.

Outcomes may not only be explainable to an agent by their own actions, but also through the actions of other agents: For instance, Bidder 1 losing an auction cycle also depends causally on the actions of the other bidders. Information flow about such causes can be specified as follows. 

\begin{definition}[External Causal Explainability]\label{def:external}
An effect $\psi$ is explainable for agent $a$ according to External Causal Explainability (ECE) in some system $\ekripke$, iff $\ekripke$ satisfies the following property:
$$\LTLglobally \big(\psi \rightarrow \exists X .\, \K_a (X \xrsquigarrow{\mathit{Act} \setminus \mathit{Act}(a)} \psi)\big) \enspace .$$
\end{definition}

Hence, ECE requires simply changing the actions that a cause must range over from $\mathit{Act}(a)$, i.e., the actions of agent~$a$, to the actions of all other agents: $\mathit{Act} \setminus \mathit{Act}(a)$.

\begin{example}
    Consider the auction system $\mathcal{A}_\mathit{explain}$ (cf.\ Figure~\ref{fig:dutch}) and the following trace:

\begin{center}
\setlength{\tabcolsep}{5.5pt}
\begin{tabular}{c | c | c | c | c | c}
 $\{o\}$ &  $\{o,b_2,e\}$ &  $\{o,b_2,b_3\}$ &  $\{o,b_1\}$ &  $\{w_2\}$ & $\{\}^\omega$
\end{tabular}
\end{center}
The temporal cause for $\lnot w_1$ at the fifth time point is now composed of Bidder 2 bidding at the second and third time point and Bidder 3 bidding at the third time point:
$$(\LTLprevious^2 ( b_2 \lor b_3 \lor \LTLprevious b_2))\, \smash{\xrsquigarrow{\{b_2,b_3,o\}}}\, \lnot w_1 \enspace .$$
\end{example}
 
Next, the combination of ICE and ECE requires full knowledge about any causal dependencies.

\begin{definition}[Full Causal Explainability]\label{def:general}
An effect $\psi$ is explainable for agent $a$ according to Full Causal Explainability (FCE) in a system $\ekripke$, iff $\ekripke$ satisfies the following property:
$$\LTLglobally \big(\psi \rightarrow \exists X .\, \K_a (X \xrsquigarrow{\mathit{Act}} \psi)\big) \enspace .$$
\end{definition}

Intuitively, adding additional atomic propositions to the causal predicates effectively requires more information to flow to agent~$a$. We can show formally this formally.

\begin{proposition}
    ICE and ECE are weaker criteria than FCE, i.e., we have for all systems $\ekripke$ that $\ekripke \models \text{FCE}$ implies $\ekripke \models \text{ICE}$ and $\ekripke \models \text{ECE}$.
\end{proposition}

\begin{proof}
    It suffices to show that for any sets $A\subseteq B$ and usual semantic parameters, (1) \smash{$\ekripke,\pi,i,\Theta \models \exists X .\, \K_a (X$ \raisebox{-1pt}{$\xrsquigarrow{B}$} $\psi)$} implies (2) \smash{$\ekripke,\pi,i,\Theta \models \exists X .\, \K_a (X$ \raisebox{-1pt}{$\xrsquigarrow{A}$} $\psi)$}. From the definition of temporal cause (cf.\ Section~\ref{sec:prelims}) it follows for any trace $\pi'$ that (3) $\mathit{Cause}(\psi,\pi',i,A) = \{ \pi'' \in \mathit{Cause}(\psi,\pi',i,B) \mid \pi' =_{(\act\setminus A)} \pi'' \}$, i.e., the cause over $A$ is exactly the subset of the cause over $B$ that is fixed over the larger action set $\act\setminus A$. From (1) we know that $\mathit{Cause}(\psi,\pi,i,B) =   \mathit{Cause}(\psi,\rho,i,B)$ for any trace $\rho$ indistinguishable to $\pi$ up to $i$, i.e., where $\Omega_a(\pi)[0,i] = \Omega_a(\rho)[0,i]$. With (3), we then also have for the same trace $\rho$ that $\mathit{Cause}(\psi,\pi,i,A) =   \mathit{Cause}(\psi,\rho,i,A)$ and (2) follows.
\end{proof}

\paragraph{Explanations Imply Knowledge.} The previously introduced definitions of explainability such as FCE (cf.\ Definition~\ref{def:general}) do not explicitly require agent $a$ to know that the effect $\psi$ holds. This is a deliberate design decision to keep the specifications succinct, as knowledge of the outcome is implied by knowledge of the counterfactual dependency, except in the case of nondeterminism on the action sequence.
\begin{proposition}\label{prop:knowledge}
    If a deterministic system $\ekripke$, a trace $\pi \in \Pi(\ekripke)$, a time point $i \in \mathbb{N}$, some $A \subseteq \AP$, and an arbitrary second-order assignment $\Theta$ satisfy
    $$\psi \land \exists X .\, \K_a (X \xrsquigarrow{A} \psi) \enspace, \enspace \text{ then also } \enspace \K_a (\psi) \enspace .$$
\end{proposition}

\begin{proof}
    From the left side of the implication we know there is a set $T$ that can instantiate $X$, such that \smash{$\K_a (X $\raisebox{-1pt}{$\xrsquigarrow{A}$} $\psi)$} holds on $\pi$ at $i$. We have assumed $\pi \models \psi$, and since $\ekripke$ is deterministic there is no $\pi''$ such that $\pi'' \leq^A_{\pi} \pi$, $\pi'' =_{(\act\setminus A)} \pi$ and $\pi'' \nmodels \psi$ (only a trace with the same action sequence could be as similar to $\pi$ as $\pi$ itself). Hence, we have $\pi|_A \in T$, where $\pi|_A[i]$ is the projection of $\pi$ to $A$ such that  $\pi|_A[i] = \pi[i]\cap A$ for all $n \in \mathbb{N}$. Now, let $\pi'$ be any trace indistinguishable to $\pi$ up to $i$, i.e., $\Omega_a(\pi)[0,i] = \Omega_a(\pi')[0,i]$. From \smash{$\K_a (X$ \raisebox{-1pt}{$\xrsquigarrow{A}$} $\psi)$} we know that $T = \mathit{Cause}(\psi,\pi',i,A) \ni \pi|_A$, and since we trivially have that $\pi' \leq^A_{\pi'} \pi|_A$ and $\pi' =_{(\act\setminus A)} \pi'$, it holds that $\pi' \models \psi$ from the definition of a temporal cause (cf.\ Section~\ref{sec:prelims}), hence the claim follows.
\end{proof}
Thus, in a deterministic system an agent can only explain present facts that they have knowledge of, while in a nondeterministic system an agent may know that some fact could be caused by nondeterminism, in which case they may not be sure whether it actually holds on a given trace.

The resolution of nondeterminism can be included more directly in explanations. This requires modeling nondeterminism as an additional agent $n$, which intuitively flips a number of coins that then determine which previously nondeterminisitc transition is taken. By including the actions $\mathit{Act}(n)$ of this agent in the respective counterfactual predicates of the explainability requirements such as ICE, an explainable system is then required to transmit information about the outcomes of these coin flips. Notably, the system resulting from such a construction is deterministic, such that Proposition~\ref{prop:knowledge} applies.

\subsection{Model Checking YLTL$^2$}\label{sec:mc}

The central challenge for model checking  \yltlso formulas is the second-order quantification ranging over sets of traces. Logics with unrestricted quantifiers of this kind are known to have an undecidable model-checking problem~\cite{BeutnerFFM23}. \yltlso restricts usage of second-order variables to causal predicates. The solution to such a causal predicate is uniquely determined for every anchor point, and this allows us to frame second-order quantification in \yltlso less as a search for a solution and more as a check for equality between the unique solutions at different anchor points. If a causal predicate appears in the scope of an epistemic or temporal operator, there may be an infinite number of such anchor points that need to be compared, which makes such a check non-trivial to achieve. 

We consider a syntactic fragment of the logic \yltlso for model checking. This fragment suffices to express all of the specifications discussed in this thesis. The fragment is called \emph{second-order conjunctive}, such that the operators in scope of any quantifier $\exists X$ are universal or conjunctive, i.e., they are only $\land$, the temporal operators $\LTLnext$, $\LTLprevious$, $\LTLglobally$ \!, or $\LTLsquareminus$ \!, and the epistemic operator $\K_a$. The fragment is also \emph{second-order flat}, such that the every right-hand subformula in a causal predicate is an LTL formula (for details, see Appendix~\ref{sec:corrigendum}).

In the proof of the following Theorem~\ref{thm:mc}, we show that second-order quantifiers in this fragment can essentially be encoded through quantification over individual traces. We express this trace quantification in a logic that can quantify over fresh atomic propositions by restricting these to follow the dynamics of the system at hand. This idea has been used in a number of related results on model-checking logics with trace quantification, i.e., so-called \emph{hyperlogics}~\cite{ClarksonFKMRS14,BozzelliMP15}. However, none of these include second-order quantifiers.

\begin{theorem}[\yltlso Model Checking]\label{thm:mc}
    There is an algorithm to decide whether a given extended transition system $\ekripke = (\kripke,\Omega)$ satisfies a second-order conjunctive and flat \yltlso formula $\varphi$, i.e., $\ekripke \models \varphi$.
\end{theorem}

\begin{proof}
    The claim follows from a reduction to the satisfiability problem of Quantified Propositional Temporal Logic (QPTL)~\cite{SistlaVW87}. QPTL extends LTL with quantifiers over fresh atomic propositions $p \in \AP$:
    $$\varphi \Coloneqq p \mid \neg \varphi \mid \varphi \lor \varphi \mid \LTLnext \varphi \mid \LTLprevious \varphi \mid \varphi \U \varphi \mid \varphi \U^- \varphi \mid \exists p . \, \varphi \enspace ,$$
    The semantics of the shared fragment are exactly as for \ltl. The propositional quantification semantics are as follows: 
    \begin{alignat*}{2}
	&\pi,i \models \exists p . \, \varphi	& \text{ iff }	& \exists \pi' \in (2^\AP)^\omega: \pi =_{\AP\setminus\{p\}} \pi' \text{ and } \pi',i \models \varphi .
\end{alignat*}
The \emph{satisfiability }problem of such a QPTL formula $\varphi$ asks whether there is a trace $\pi \in (2^\AP)^\omega$ such that $\pi,0 \models \varphi$ and is decidable~\cite{SistlaVW87}. 

We can reduce the \yltlso model checking to QPTL satisfiability via a translation function $\mathit{enc}$ that encodes the extended transition system $\ekripke = (\kripke,\Omega)$ and the \yltlso formula $\varphi$ into a QPTL formula. We combine several translations from Bozzelli, Maubert, and Pinchinat~\shortcite{BozzelliMP15} for the \kltl-fragment of \yltlso. The second-order quantification requires a novel, non-trivial approach that we outline afterward. The translation models different paths of the transition system through distinct atomic propositions. Therefore, the resulting formula ranges over an augmented set $\AP' = \{ \, p_\pi \mid p~\in(\AP\cup~S) \land \pi \in \text{PV}\}$, where $\text{PV}$ denotes a set of \emph{path variables}. For a \yltlso formula $\varphi$, it suffices to introduce one such path variable for every knowledge operator $\K_a$ and causal predicate, two for every second-order quantifier, and one initial variable $\alpha$ that encodes the universal application of the trace semantics. Hence, $\AP'$ is finite. We can enforce that these new propositions in $\AP'$ evolve according to the transitions of $\ekripke = (\kripke,\Omega)$:
\begin{align*}
\theta(\pi,\kripke) = &\LTLpastfinally \Big((\lnot\LTLprevious\top) \land (s_0)_\pi\land \LTLglobally \bigwedge_{s\in S}  \big[ s_\pi \rightarrow  \bigwedge_{t\in S \setminus \{s\}} \!\!\!\!\lnot t_\pi \\ & \land \bigvee_{t\in \Delta(s)} \!\!\big( \LTLnext t_\pi \land  \!\bigwedge_{p\in \Lambda(s,t)} p_\pi \land \!\bigwedge_{p\in \AP \setminus \Lambda(s,t)} \!\!\!\!\lnot p_\pi   \big)\big] \Big) \, .
\end{align*}
In the end, our proof establishes the following equivalence, where $\theta(\alpha,\kripke)$ is used to quantify over all initial traces of $\kripke$: The model $\ekripke$ satisfies the formula $\varphi$ iff
\begin{align}
    \forall \AP_\alpha \cup S_{\alpha}. \, \theta(\alpha,\kripke) \rightarrow \mathit{enc}(\alpha,\varphi,v_\mathfrak{E})\label{eq:univquant}
\end{align}
is satisfiable. The third parameter of $\mathit{enc}$ is a mapping from second-order variables to trace variables, where $v_\mathfrak{E}$ denotes the empty mapping.
The third parameter will play a central role later in the translation of second-order quantifiers, which we discuss after the simple cases.\smallskip

\emph{Temporal Operators:} Since QPTL has the same temporal operators as \yltlso, the translation in these cases is straightforward. We define $\mathit{enc}$ through recursion on $\varphi$:
\begin{alignat*}{2}
	&\mathit{enc}(\pi,p,v) & &= p_\pi\\
	&\mathit{enc}(\pi,\lnot \psi,v) & &= \lnot \mathit{enc}(\pi,\psi,v)\\
	&\mathit{enc}(\pi,\psi_1 \land \psi_2,v) & &= \mathit{enc}(\pi,\psi_1,v) \land \mathit{enc}(\pi,\psi_2,v)\\
	&\mathit{enc}(\pi,\LTLnext \psi,v) & &= \LTLnext \mathit{enc}(\pi,\psi,v) \enspace.
\end{alignat*}
The other operators follow analogously. Intuitively, $\pi$ is the currently scoped trace variable and corresponds to $\pi$ in the semantics of KLTL and \yltlso (cf.\ Sections~\ref{sec:prelims} and~\ref{sec:synsem}).\smallskip

\emph{Epistemic Operator:} For the epistemic operator $\K_a$, the translation exploits that quantification over paths $\pi \in \kripke$ can be encoded through QPTL quantifiers as outlined already with Formula~\ref{eq:univquant}. 
We additionally need to constrain these initial paths to be observation equivalent for agent~$a$ as follows:
\begin{alignat*}{2}
	&\mathit{enc}(\pi,\K_a \, \psi,v) = \, & \forall \AP_\rho \cup S_\rho.&\big(\theta(\rho,\kripke) \land \LTLpastglobally(\pi =_{\Omega(a)} \rho)\big)\\ & & & \rightarrow \mathit{enc}(\rho,\psi,v)\enspace ,
\end{alignat*}
where observation equivalence $=_{\Omega(a)}$ for an agent $a$ at a certain time point is translated as follows:
$$ \pi =_{\Omega(a)} \rho \equiv \bigwedge_{p \in \Omega(a)} p_\pi \leftrightarrow p_\rho \enspace .$$

\emph{Second-Order Quantifiers \& Causal Predicates:} The second-order quantifiers for causes require a more complex encoding than $\K_a$, since they quantify over sets of traces and not single traces. The main idea of our encoding is that for any existentially quantified causal set, all initial system traces have to either satisfy \emph{all} associated causal predicates or \emph{none} of them. In the former case the trace is in the cause at all anchor points, while in the latter case the trace is in the complement at all anchor points. This ensures that there are no traces that are in the cause at some anchor points and not in the cause at others, which would mean these causes are not equal and no single set qualifies at every anchor point. Note that this connection exploits that causes are uniquely determined at every anchor point, if they exist. We can encode these requirements with:
\begin{alignat}{2}
	&\mathit{enc}(\pi,\exists X. \, \psi,v) = \, & &\big(\forall \AP_\rho \cup S_\rho. \mathit{enc}(\pi, \psi,v_1)\label{eq:incause}\\
    & & &  \lor \mathit{enc}(\pi, \psi,v_2)\big)\enspace .\label{eq:outcause}
\end{alignat}
The variable mappings $v_1$ and $v_2$ encode whether the subformula refers to the causal set or its complement, respectively: 
\begin{align*}
	&v_1(Y) = \,\begin{cases}(\rho,\top)~&\text{if Y = X,}\\ v(Y)~&\text{otherwise, and} \end{cases}\\
	&v_2(Y) = \, \begin{cases}(\rho,\bot)~&\text{if Y = X,}\\ v(Y)~&\text{otherwise.} \end{cases}
\end{align*}
Intuitively, Subformula~\ref{eq:incause} encodes that the trace of $\rho$ is in the cause $X$ at all causal predicates \smash{$X$ \raisebox{-1pt}{$\xrsquigarrow{A}$} $\varphi$}, while Subformula~\ref{eq:outcause} encodes that the trace is in no such cause. The mappings $v_1$ and $v_2$ track the location of subformulas, i.e., by mapping second-order variables to $\top$ and $\bot$, respectively. It remains to encode the (non-)membership in $X$ at the location of the causal predicates: 
\begin{alignat*}{2}
	&\mathit{enc}(\pi,X \xrsquigarrow{A} \varphi,v) =  \begin{cases}\psi_\mathit{cause}~&\text{if }v(X).2 = \top,\\ \lnot \psi_\mathit{cause}~&\text{if }v(X).2 = \bot,
    \end{cases}
\end{alignat*}
Where $v(X).2$ stands for the second component of the tuple (and $v(X).1$ for the first), and 
where we encode membership in $\mathit{Cause}(\psi,\pi,i,A)$ via quantifying over individual traces to express that $v(X).1$ is in the downward-closed set as defined in Section~\ref{sec:prelims}:
\begin{alignat*}{2}
    \psi_\mathit{cause} = &\forall \AP_\sigma \cup S_\sigma. \theta(\sigma,\kripke) \rightarrow \\ &\big(\sigma \leq^A_\pi v(X).1 \rightarrow \mathit{enc}(\sigma,\varphi,v)\big) \enspace .
\end{alignat*}
The similarity relation $\leq^A_\pi$ can be translated into temporal logic as follows:
\begin{align*}
    \sigma \leq^A_\pi \rho \equiv \LTLpastfinally \Big(&(\lnot \LTLprevious\top) \land \LTLglobally \big[\bigwedge_{p \in \act \setminus A} p_\sigma \leftrightarrow p_\pi\\
    &\land \bigwedge_{p \in A} (p_\sigma \not\leftrightarrow p_\pi) \rightarrow (p_\rho \not\leftrightarrow p_\pi)\big]\Big)\enspace .
\end{align*}
This concludes the description of $\mathit{enc}$. The equivalence of the model-checking problem to the satisfiability of Formula~\ref{eq:univquant} can be shown through structural induction on $\varphi$.
\end{proof}

\paragraph{Complexity.} It is easy to see that the translation function $\mathit{enc}$ may double the size of the formula when encoding a second-order quantifier (cf. Equations~\ref{eq:incause} and~\ref{eq:outcause}). 
Less obvious is an added propositional quantifier alternation, introduced by translating causal predicates \smash{$X$ \raisebox{-1pt}{$\xrsquigarrow{A}$} $\varphi$} in the second disjunct (cf. Equation~\ref{eq:outcause}); these predicates are translated to $\lnot \psi_\mathit{cause}$ by adding a negated universal quantifier. 
This results in a non-elementary (tower-exponential) space complexity of model-checking \yltlso formulas, where the tower grows with any nesting of second-order quantifiers in the scope of causal predicates and with nesting of knowledge operators and negations.
Previous results suggest that this complexity is close to optimal: \yltlso subsumes \kltl, which has a similar non-elementary lower bound scaling with nested epistemic operators~\cite{BozzelliMM24}. 
Constructing temporal causes as automata has a doubly exponential lower bound, in the size of the consequent~\cite{CarelliFS25}, which matches with nested causal predicates growing the exponents of the model-checking complexity.

In practice, the specifications we considered, such as ICE, ECE and FCE, all only contain one propositional quantifier alternation in their encodings, and model-checking them with the algorithm outlined for Theorem~\ref{thm:mc} therefore scales exponentially in the size of the effect $\varphi$ and polynomially in the size of the system. Our experiments confirm that these specifications can be verified for systems of nontrivial size.

\section{Experiments}\label{sec:experiments}

We report experiments on verifying explainability and privacy requirements specified in \yltlso. Although we describe an encoding into QPTL satisfiability in Theorem~\ref{thm:mc} for brevity, there is no satisfiability checker for QPTL. Hence, our evaluation uses the model-checking tool AutoHyper~\cite{BeutnerF23a} as a backend procedure, because QPTL satisfiability can be encoded in HyperLTL model checking~\cite{FinkbeinerRS15}. Details can be found in Appendix~\ref{app:enc}. The evaluation was conducted on MacOS with an M3 Pro 4.05~GHz processor and 36GB of memory.\footnote{Code: \url{https://doi.org/10.5281/zenodo.16421482}}

\begin{table}[t]
		\centering
            \def\arraystretch{1.1}
		\setlength\tabcolsep{1.5mm}
		\setlength\dashlinedash{1pt}
		\setlength\dashlinegap{2pt}
		\begin{tabular}{lccccc}
			\toprule
			\textbf{Type} & $|B|$ & ICE & ECE & FCE & Priv.~I \\
			\midrule
			\textsc{Blind}  & 
            \makecell{2 \\  3 \\  4 \\ 5} & 
            \makecell{\xmark /0.85 \\  \xmark /1.57 \\  \xmark /3.42 \\ \xmark /7.85} & 
            \makecell{\xmark /0.88 \\  \xmark /1.77 \\  \xmark /3.68 \\ \xmark /9.45} & 
            \makecell{\xmark /0.96 \\  \xmark /1.73 \\  \xmark /3.61 \\ \xmark /10.3} & 
            \makecell{\cmark /0.25 \\  \cmark /0.27 \\  \cmark /0.27 \\ \cmark /0.22} \\
            \midrule
			\textsc{Public}  & 
            \makecell{2 \\  3 \\  4 \\ 5} & 
            \makecell{\cmark /0.78 \\  \cmark /1.47 \\  \cmark /3.19 \\ \cmark /7.24} & 
            \makecell{\cmark /0.87 \\  \cmark /1.80 \\  \cmark /3.73 \\ \cmark / 9.52} & 
            \makecell{\cmark /0.81 \\  \cmark /1.63 \\  \cmark /3.48 \\ \cmark /10.3} & 
            \makecell{\xmark /0.30 \\  \xmark /0.30 \\  \xmark /0.50 \\ \xmark /1.21} \\
            \midrule
			\textsc{Explain}  & 
            \makecell{2 \\  3 \\  4 \\ 5} & 
            \makecell{\cmark /0.86 \\  \cmark /1.39 \\  \cmark /3.40 \\ \cmark /8.03} & 
            \makecell{\xmark /1.02 \\  \xmark /1.89 \\  \xmark /3.82 \\ \xmark / 9.92} & 
            \makecell{\xmark /1.02 \\  \xmark /1.82 \\  \xmark /4.03 \\ \xmark /10.4} & 
            \makecell{\xmark /0.24 \\  \cmark /0.25 \\  \cmark /0.29 \\ \cmark /0.29} \\
			\bottomrule
		\end{tabular}
\caption{Verification results and runtime in seconds for checking the Dutch auction model introduced in Section~\ref{sec:motivation} with a set of bidders $B$ against the explainability requirements introduced in Sections~\ref{sec:reqs} and $b_2$-privacy (Priv.~I, cf.~Section~\ref{sec:motivation}).}\label{tab:dutch}
\end{table}

\subsection{Dutch Auction}

We conducted experiments with all three versions of the Dutch auction system introduced in Section~\ref{sec:motivation} and the explainability requirements introduced in Secion~\ref{sec:reqs}. We also checked $b_2$-privacy as introduced in Section~\ref{sec:motivation}. In Table~\ref{tab:dutch}, we list whether the respective auction version satisfies a requirement and how much time the model checker required to produce this verdict. We consider a scaling number of bidders $B$, which increases the size of the transition system exponentially, and this scaling is mirrored in the runtimes.

The verification verdicts match the intuition described in Section~\ref{sec:motivation}: The blind auction is not explainable but private, while the public auction naturally ensures sufficient flow of information to achieve all three explainability requirements, but also reveals the private bids. The explainable auction transports enough information on the causal dependencies over the actions of the bidders themselves, but does not ensure ECE or FCE. An interesting case is the privacy of the explainable auction in the case of two bidders. In this scenario, revealing the value of the highest bid allows an agent to infer that the other agent placed the bid in case they did not bid themselves, in this way violating the privacy constraint. Hence, explainability and $b_2$-privacy can only be achieved with more than two bidders. 

\subsection{Rock Paper Scissors}

To differentiate ICE and ECE in more detail, we model the classic game of Rock Paper Scissors, where every round two players select between the three eponymous objects  ($r$, $p$, and $s$). Matching objects result in a draw, otherwise $p$ beats $r$, $s$ beats $p$, and $r$ beats $s$. An agent $i$ does not see the selection of the other, but only the outcome draw~($d$) or loss~($l_i$). We consider explainability of $l_1$, i.e., the loss of Agent 1. Since an agent knows what action their picked one looses to, and which action would have won instead, even this blind version of the game is fully explainable according to all three requirements. The verification times are as follows.

\begin{table}[!h]
\centering
            \def\arraystretch{1.1}
		\setlength\tabcolsep{2mm}
\begin{tabular}{lccccc}
			\toprule
			\textbf{Type} & ICE & ECE & FCE & Priv.~II \\
			\midrule
			\textsc{Standard} & \cmark /0.49 & \cmark /0.53 & \cmark /0.55  & \xmark /0.24 \\
            \midrule
			\textsc{ + Well}  & \xmark /1.02 & \cmark /0.93 & \xmark /1.08 &  \cmark /0.29  \\
			\bottomrule
\end{tabular}
\end{table}
\noindent We also consider a popular variant with a fourth action, Well ($w$), which beats $s$ and $r$ but looses to $p$. This version does not satisfy ICE because, e.g., when loosing with rock ($r_1$) Agent 1 does not know whether scissors (against $p_2$) or paper (against $w_2$) would have been their winning move. ECE is still satisfied because the agent knows against which moves they would have won, i.e., $p_1$ would have won if the other agent had played $r_2$ or $w_2$ (the cause is then $\lnot r_2 \land \lnot w_2$). We also verified the following conditional privacy specification (results reported under Priv.~II): $ \LTLglobally (\lnot d \rightarrow  \lnot \K_1(r_2))$, i.e., whenever the outcome is not a draw, an agent does not know that the other agent played $r$. This is only satisfied in the extended version of the game, because $w_2$ and $s_2$ can simulate $r_2$, except in case of a draw.

\subsection{Matching Pennies}

To further explore the scalability of our algorithm we consider a blind version of the game Matching Pennies, played collaboratively: Each player chooses heads or tails for their coin and all players win together when their choices match. A player $i$ only sees their coin $c_i$ and the outcome $w$, as well as an explanatory \emph{blaming} output $b_i$, which is enabled if their coin was the only mismatch. Without the blaming output, the setup satisfies all three explainability requirements only in the 2-player case. With the blaming output, it additionally satisfies ICE for any number of players. Runtime results for the latter experiments are shown in Figure~\ref{fig:pennies}. They confirm that model-checking the three explainability requirements scales polynomially  with an increasing system size. The plot also shows that the time needed to verify the conditional privacy requirement $ \LTLglobally (\lnot w  \rightarrow \lnot \K_1(c_2))$ stays practically  constant (cf.~Priv.\ III). This property states that when Player 1 does not win, they do not know the value of Player 2's coin. This is satisfied only by the non-blaming version with more than two players. We suspect this property allows the model checker to fully abstract away from the coins of the other players, while this is not possible for the explainability specifications, which place constraints on these coins in the encoding of $\leq^A_\pi$ (cf.\ proof of Theorem~\ref{thm:mc}).

\begin{figure}[t]
\centering
\includegraphics[width=0.45\textwidth]{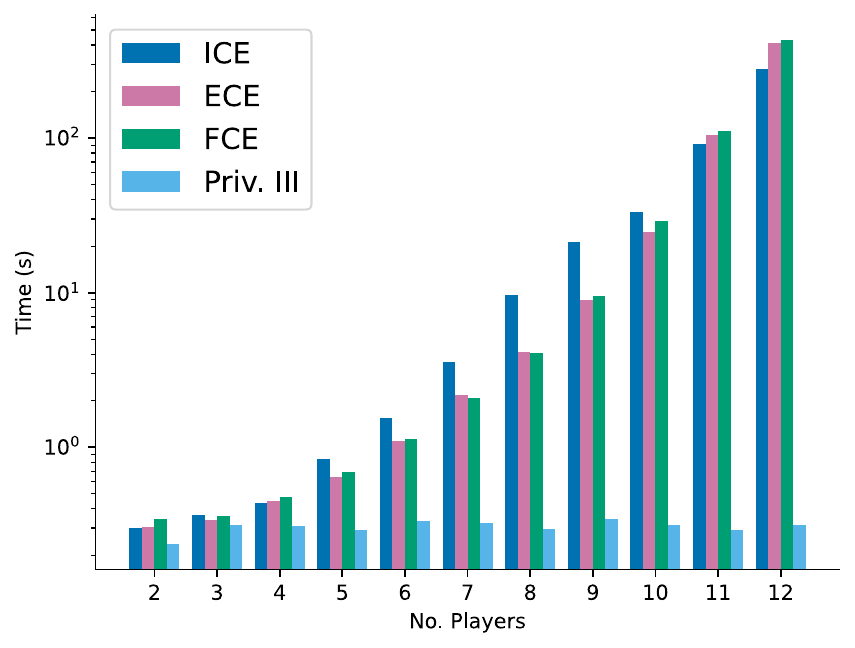}
\caption{Runtime results for verifying explainability and privacy requirements of blind Matching Pennies games with blaming output. Note that system size scales exponentially with the number of players, i.e., the 12-player game has 4096 transitions.}\label{fig:pennies}
\end{figure}

\section{Related Work}\label{sec:related}

\textit{Explanation Generation:} Our work is concerned with formalizing explainability as a system property and not with generating explanations like other literature on explainable AI~\cite{Ribeiro0G16,LundbergL17,ShihCD18}. Nevertheless, these generation methods are a central motivation for us because the systems resulting from their implementation will need to tread the fine line between explainability and privacy that we study here. Applying our insights may be easier for model-based approaches that deal with descriptive systems such as decision trees~\cite{DarwicheJ22,DBLP:conf/nips/ArenasBBPS21,DBLP:conf/kr/ArenasBBCS24,CarbonnelC023} than for black-box methods based on, e.g., abduction~\cite{IgnatievNM19} or heuristics~\cite{AngelinoLASR17,LiLCR18}.

\textit{Formal Explanations:} Logical reasoning techniques to compute explanations are widely studied~\cite{Darwiche23,Wu0B23,LeofanteBR23}. Khan and Lespérance~\shortcite{KhanL21} combine epistemic and causal reasoning in the situation calculus, which has been extended to explain agent behavior~\cite{KhanR23}. Temporal properties are used as explanations in AI planning~\cite{KimMSAS19,EiflerST020}; another popular framework in planning is model reconciliation~\cite{ChakrabortiSZK17}. Planning with explanatory actions~\cite{ChakrabortiKSSK19,Sreedharan19} follows a similar goal as us in effecting the epistemic state through observations. Several works on axiomatizing explainable classifiers~\cite{AmgoudB22,LiuL23} employ counterfactual reasoning and consider partial knowledge. There are works on actual causality~\cite{Halpern16} for explanations~\cite{ChocklerH24}, as well as counterfactual modal logic~\cite{Aguilera-Ventura23}, but we use temporal causes~\cite{CoenenFFHMS22,FinkbeinerFMS24} to deal with the sequential and possibly non-terminating nature of multi-agent systems.\smallskip

\textit{Logics for Information Flow:} Logics for hyperproperties~\cite{ClarksonS10} have been a popular framework to study information flow and partly subsume epistemic temporal logic~\cite{BozzelliMP15,Rabe16}, which has similarly been used for information-flow control~\cite{BalliuDG11,HalpernO08}. Counterfactual reasoning has been encoded in such \emph{hyperlogics}~\cite{CoenenFFHMS22,FinkbeinerS23,BeutnerFFS23} for checking hypotheses. Several works study hyperlogics in multi-agent systems without second-order quantifiers~\cite{BeutnerF23,BeutnerF24,BeutnerF24b}. The closest work to ours considers a second-order hyperlogic without decidable model checking~\cite{BeutnerFFM23,BeutnerFF024}.\smallskip

\textit{Interdisciplinary Aspects:} Our perspective on explainability is rooted in logic and information-flow theory. Building self-explanatory systems from explainable system models is an intriguing interdisciplinary problem that lies outside the scope of this paper. There are user studies regarding which explanations are preferred by users and how to visualize them~\cite{SeimetzE021,SchlickerLOBKW21,HorakCMHFMDFD22,EiflerBCF022,BrandaoMMLC22,DelaneyPGK23}. We use temporal causes, which capture the idea of minimally editing previous actions of the agents, but our information-flow perspective may also be applied to other explanantia such as presenting counterfactual  executions. This requires modifying the encoding in the proof of Theorem~\ref{thm:mc}. High-level taxonomies for explainability requirements~\cite{KohlBLOSB19,LangerOSHKSSB21} have inspired us to concretize them with formal logic and multi-agent systems. 

\section{Summary \& Conclusion}\label{sec:summary}

This paper presents a logic for multi-agent systems to formally specify explainability requirements of the form: when a certain outcome happens, an agent knows \emph{why}, i.e., what actions caused the outcome. This is expressed with a combination of counterfactual, epistemic, and temporal operators, as well as second-order quantification over sets of  executions. Privacy requirements can be encoded in the same logic and we have described an algorithm to automatically verify whether a system model satisfies such specifications.

Our theoretical and experimental results suggest that our formal explainability specifications capture what it means for a system to be explainable in a qualitative information-theoretic sense: An explainable system needs to ensure sufficient flow of information about causes via a number of observations that serve as explanations. On an abstract level, this means that in explainable systems observation-equivalence needs to be finer than causality.

We discovered an inherent tradeoff between explainability and privacy, which is an intriguing avenue for future work. Automatic satisfiability checking of our logic may allow to identify general rules for this tradeoff beyond analyzing a given system, while repair algorithms may add a minimal set of explanatory observations to a system model to achieve explainability without sacrificing privacy.

\section*{Acknowledgements}
This work was partially supported by the DFG in project 389792660 (TRR 248 -- CPEC) and funded by the European Union through ERC Grant HYPER (No.\ 101055412). Views and opinions expressed are however those of the authors only and do not necessarily reflect those of the European Union or the European Research Council Executive Agency. Neither the European Union nor the granting authority can be held responsible for them.
This work was supported in part by The Israel Science Foundation (grant No.\ 655/25).

\bibliographystyle{kr}
\bibliography{references}

\newpage
\onecolumn

\appendix

\section{Corrigendum}\label{sec:corrigendum}

The preliminary version of this paper wrongfully claimed that the model-checking encoding is sound for the full logic. However, it is only sound for the fragment described in this extended and corrected version of the paper. This is because disjunctions allow to build formulas with multiple valid solutions for a second-order variable. For these formulas, the encoding leads to wrong inferences, as outlined in the following example.

\begin{example}
    Consider the set of traces $T = \{\{a\}\{\}^\omega,\{b\}\{\}^\omega,\{a,b\}\{\}^\omega,\{\}\{\}^\omega\}$, the set $A = \{a,b\}$, and the \yltlso formula 
    $$\varphi = \exists X . \big(X \xrsquigarrow{A} (a \lor b)\big) \land \big((X \xrsquigarrow{A} a) \lor (X \xrsquigarrow{A} b)\big) \enspace .$$
    The system defined by $T$ does not satisfy the formula, as the cause for $a \lor b$ is defined by the prefixes $\{\{a\}\ldots,\{b\}\ldots,\{a,b\}\ldots\}$ on all traces that satisfy $a$ or $b$ and by $\emptyset$ on $\{\}\{\}^\omega$,
    the cause for $a$ is defined by
    $\{\{a\}\ldots,\{a,b\}\ldots\} \enspace $ on all traces that satisfy $a$ and by $\emptyset$ on all others,
    and lastly the cause for $b$ is defined by
    $\{\{b\}\ldots,\{a,b\}\ldots\} \enspace $ on all traces that satisfy $b$ and by $\emptyset$ on all others.
    Hence, e.g., on $\{a,b\}\{\}^\omega$ there is no way to instantiate $X$ such that is both the cause of $a \lor b$ and either $a$ or $b$. However, the system does satisfy the encoding of the formula. This is because the encoding corresponds to the following satisfied requirement:
    \begin{align*}
        \forall \pi \in T . \forall \pi' \in (2^A)^\omega . \, &\big(\pi' \in \mathit{Cause}(a \lor b,\pi,0,A) \land (\pi' \in \mathit{Cause}(a,\pi,0,A) \lor \pi' \in \mathit{Cause}(b,\pi,0,A))\big) \lor\\ &\big(\pi' \notin \mathit{Cause}(a \lor b,\pi,0,A) \land (\pi' \notin \mathit{Cause}(a,\pi,0,A) \, \lor \pi' \notin \mathit{Cause}(b,\pi,0,A))\big) \enspace .
    \end{align*}
    Intuitively, the disjunctions now allow individual traces to choose between two solutions for $X$ required by their subformulas. Consequently, a disjunction such as $\mathit{Cause}(a,\pi,0,A) \lor \pi' \in \mathit{Cause}(b,\pi,0,A)$ is treated as the union of the two causes.
\end{example}

Overall, the restriction to second-order conjunctive and flat subformulas has no consequence for the other claims in the paper, as all presented explainability and privacy requirements can be expressed in the described fragment, which is defined formally in the following.

\begin{definition}[Second-Order Conjunctive and Flat \yltlso Formulas]\label{frag}
The fragment of second-order conjunctive and flat \yltlso formulas is defined formally by the following grammar:
\begin{align*}
	\varphi &\Coloneqq \; p  \mid \neg \varphi \mid \varphi \land \varphi \mid \LTLnext \varphi \mid \varphi \LTLuntil \varphi \mid \LTLprevious \varphi \mid \varphi \U^- \varphi \mid\K_a \, \varphi \mid \exists X . \, \psi \enspace ,\\
    \psi &\Coloneqq \; p  \mid \psi \land \psi \mid \LTLnext \psi \mid \LTLprevious \psi \mid \LTLglobally \psi \mid \LTLpastglobally \psi \mid \K_a \, \psi \mid X \xrsquigarrow{A} \eta \enspace ,
\end{align*}
where $p \in \AP$ and $\eta$ is an LTL formula~\cite{Pnueli77}.
\end{definition}

\section{HyperQPTL Encoding and Structural Induction}\label{app:enc}

Here, we outline the encoding into HyperQPTL~\cite{Rabe16} used for the experiments. We first encoded \yltlso formulas into non-prenex HyperQPTL formulas with past operators built according to the following grammar:
\begin{align*}
\varphi &\Coloneqq \ap_\trv \mid p \mid \neg \varphi \mid \varphi \land \varphi \mid \LTLnext \varphi \mid \varphi \LTLuntil \varphi  \mid \LTLprevious \varphi \mid \varphi \U^- \varphi \mid \forall \trv \ldot \varphi \mid \forall \ap \ldot \varphi \enspace .
\end{align*}
Here, $\trv \in \tracevars$ is a trace variable and $\ap \in \AP$ are atomic propositions. Formulas in this form can be converted to prenex form in the same way as QPTL formulas~\cite{SistlaVW87}, and then past operators can be removed to obtain a pure-future LTL body~\cite{LichtensteinPZ85}. The semantics of the formulas is given as follows, where $\TR$ is a set of traces and $\traceassign: \tracevars \rightarrow (2^\AP)^\omega$ a trace assignment that maps trace variables to traces. 
We write $\traceassign[\trv \mapsto \tr]$, to update $\traceassign$ such that it maps trace variable $\trv$ to trace $\tr$.
\begin{equation*}
\begin{array}{lll}
\traceassign,i \models_\TR \ap_\trv       \qquad \qquad & \text{iff } & \ap \in \traceassign(\trv)[i]\\
\traceassign,i \models_\TR p     \qquad \qquad & \text{iff } &  p \in \traceassign(\trv_p)[i] \\
\traceassign,i \models_\TR \neg \varphi              & \text{iff } & \traceassign,i \not\models_\TR \varphi \\
\traceassign,i \models_\TR \varphi \land \psi         & \text{iff } & \traceassign,i \models_\TR \varphi \text{ and } \traceassign,i \models_\TR \psi \\
\traceassign,i \models_\TR \LTLnext \varphi                & \text{iff } & \traceassign,i+1 \models_\TR \varphi \\
\traceassign,i \models_\TR \varphi\LTLuntil\psi             & \text{iff } & \exists k \geq i \text{ such that } \traceassign,k \models_\TR \psi \text{ and } \forall i \leq j < k \text{ it holds that } \traceassign,j \models_\TR \varphi\\
\traceassign,i \models_\TR \LTLprevious \varphi             & \text{iff } & \traceassign,i-1 \models_\TR \varphi \text{ and } i > 0 \\
\traceassign,i\models_\TR \varphi\U^-\psi             & \text{iff } & \exists k\leq i \text{ such that } \traceassign,k\models_\TR \psi \text{ and } \forall k < j \leq i \text{ it holds that } \traceassign,j\models_\TR \varphi\\
\traceassign,i \models_\TR \forall \trv \ldot \varphi & \text{iff } & \forall \tr \in \TR \text{ it holds that } \traceassign[\trv \mapsto \tr],i \models_\TR \varphi\\
\traceassign,i \models_\TR \forall p \ldot \varphi & \text{iff } & \forall \tr_p \in (2^{\{p\}})^\omega \text{ it holds that } \traceassign[\trv_p \mapsto \tr_p],i \models_{\TR} \varphi
\end{array}
\end{equation*}
System-level satisfaction is defined as for HyperQPTL: It requires $\emptyset, 0 \models_{\traces(\kripke)}\varphi$.

\begin{definition}[HyperQPTL Encoding]\label{def:hyperenc}
    We outlined the encoding into HyperQPTL, which was used for the experiments and directly corresponds to the encoding into QPTL described in Theorem~\ref{thm:mc}.
    \begin{alignat*}{2}
	&\mathit{enc}(\trv,\ap,v) & &= \ap_\trv\\
	&\mathit{enc}(\trv,\lnot \psi,v) & &= \lnot \mathit{enc}(\trv,\psi,v)\\
	&\mathit{enc}(\trv,\psi_1 \land \psi_2,v) & &= \mathit{enc}(\trv,\psi_1,v) \land \mathit{enc}(\trv,\psi_2,v)\\
	&\mathit{enc}(\trv,\LTLnext \psi,v) & &= \LTLnext \mathit{enc}(\trv,\psi,v)\\
	&\mathit{enc}(\trv,\LTLprevious \psi,v) & &= \LTLprevious \mathit{enc}(\trv,\psi,v)\\
	&\mathit{enc}(\trv,\psi_1 \LTLuntil \psi_2,v) & &= \mathit{enc}(\trv,\psi_1,v) \LTLuntil \mathit{enc}(\trv,\psi_2,v)\\
	&\mathit{enc}(\trv,\psi_1 \U^- \psi_2,v) & &= \mathit{enc}(\trv,\psi_1,v) \U^- \mathit{enc}(\trv,\psi_2,v)\\
	&\mathit{enc}(\trv,\K_a \, \psi,v) & &= \, \forall \rho. \big(\LTLpastglobally(\bigwedge_{p \in \Omega(a)} p_\rho \leftrightarrow p_\trv) \rightarrow \mathit{enc}(\rho,\psi,v) \big)\\
	&\mathit{enc}(\trv,\exists X. \, \psi,v) & &= \, \big(\forall \rho. \mathit{enc}(\trv, \psi,v[X \mapsto (\rho,\top)]) \lor \mathit{enc}(\trv, \psi,v[X \mapsto (\rho,\top)])\big)\\
	&\mathit{enc}(\trv,X \xrsquigarrow{A} \varphi,v) & &=  \begin{cases}\forall \sigma. \big( \sigma \leq^\mathit{A}_\trv v(X).1 \rightarrow \mathit{enc}(\sigma,\varphi,v) \big)~&\text{if }v(X).2 = \top,\\ \exists \sigma. \big( \sigma \leq^\mathit{A}_\trv v(X).1 \land \mathit{enc}(\sigma,\varphi,v) \big)~&\text{if }v(X).2 = \bot,
    \end{cases}
    \end{alignat*}
    
\end{definition}

In the remainder of this section we prove the correctness of the HyperQPTL encoding. This also indirectly proves the correctness of the encoding into QPTL satisfiability describe din Theorem~\ref{thm:mc} due to the correspondence between HyperQPTL model checking and QPTL satisfiability~\cite{Rabe16}.

\begin{theorem}\label{thm:yltlsomc}
    We show the following equivalence for a well-formed \yltlso formula $\varphi$ of the fragment defined in Definition~\ref{frag} and some set of traces $T$, trace $\tr \in T$ and time point $i \geq 0$, using the encoding function $\mathit{enc}$ as described in Definition~\ref{def:hyperenc}.
\begin{align*}
    \ekripke,\tr,i,\mathfrak{E} \models \varphi \text{ iff } \emptyset[\trv \mapsto \tr],i \models_{\Pi(\ekripke)} \mathit{enc}(\trv,\varphi,v_\mathfrak{E}) \enspace .
\end{align*}
\end{theorem}

\begin{proof}

Via structural induction over $\varphi$.\\
\textbf{Case} $\varphi = \ap$: Follows directly form the semantics of \yltlso and HyperQPTL.

\smallskip
\noindent\textbf{Case} $\varphi = \lnot \psi$: The equivalence can be established as follows. $\ekripke,\tr,i,\mathfrak{E} \models \lnot \psi$ iff (semantics of \yltlso) $\ekripke,\tr,i,\mathfrak{E} \not\models \psi$ iff (induction hypothesis) $\emptyset[\trv \mapsto \tr],i \not\models_T \mathit{enc}(\trv,\psi,v_\mathfrak{E})$ iff (semantics of HyperQPTL) $\emptyset[\trv \mapsto \tr],i \models_{\Pi(\ekripke)} \lnot \mathit{enc}(\trv,\psi,v_\mathfrak{E})$ iff (definition of $\mathit{enc}$) $\emptyset[\trv \mapsto \tr],i \models_{\Pi(\ekripke)} \mathit{enc}(\trv,\lnot \psi,v_\mathfrak{E})$.

\smallskip
\noindent\textbf{Case} $\varphi = \varphi_1 \land \varphi_2$: The equivalence can be established as follows. $\ekripke,\tr,i,\mathfrak{E} \models \varphi_1 \land \varphi_2$ iff (semantics of \yltlso) $\ekripke,\tr,i,\mathfrak{E} \models \varphi_1$ and $\ekripke,\tr,i,\mathfrak{E} \models \varphi_2$ iff (induction hypothesis) $\emptyset[\trv \mapsto \tr],i \models_{\Pi(\ekripke)} \mathit{enc}(\trv,\varphi_1,v_\mathfrak{E})$ and $\emptyset[\trv \mapsto \tr],i \models_{\Pi(\ekripke)} \mathit{enc}(\trv,\varphi_2,v_\mathfrak{E})$ iff (semantics of HyperQPTL) $\emptyset[\trv \mapsto \tr],i \models_{\Pi(\ekripke)} \mathit{enc}(\trv,\varphi_1,v_\mathfrak{E}) \land \mathit{enc}(\trv,\varphi_2,v_\mathfrak{E})$ iff (definition of $\mathit{enc}$) $\emptyset[\trv \mapsto \tr],i \models_{\Pi(\ekripke)} \mathit{enc}(\trv,\varphi_1 \land \varphi_2,v_\mathfrak{E})$.

\smallskip
\noindent\textbf{Case} $\varphi = \LTLnext \psi$: The equivalence can be established as follows. $\ekripke,\tr,i,\mathfrak{E} \models \LTLnext \psi$ iff (semantics of \yltlso) $\ekripke,\tr,i+1,\mathfrak{E} \models \psi$ iff (induction hypothesis) $\emptyset[\trv \mapsto \tr],i+1 \models_{\Pi(\ekripke)} \mathit{enc}(\trv,\psi,v_\mathfrak{E})$ iff (semantics of HyperQPTL) $\emptyset[\trv \mapsto \tr],i \models_{\Pi(\ekripke)} \LTLnext \mathit{enc}(\trv,\psi,v_\mathfrak{E})$ iff (definition of $\mathit{enc}$) $\emptyset[\trv \mapsto \tr],i \models_{\Pi(\ekripke)} \mathit{enc}(\trv,\LTLnext \psi,v_\mathfrak{E})$.

\smallskip
\noindent\textbf{Case} $\varphi = \LTLprevious \psi$: The equivalence can be established as follows. $\ekripke,\tr,i,\mathfrak{E} \models \LTLprevious \psi$ iff (semantics of \yltlso) $\ekripke,\tr,i-1,\mathfrak{E} \models \psi$ and $i > 0$ iff (induction hypothesis) $\emptyset[\trv \mapsto \tr],i-1 \models_{\Pi(\ekripke)} \mathit{enc}(\trv,\psi,v_\mathfrak{E})$ and $i > 0$ iff (semantics of HyperQPTL) $\emptyset[\trv \mapsto \tr],i \models_{\Pi(\ekripke)} \LTLprevious \mathit{enc}(\trv,\psi,v_\mathfrak{E})$ iff (definition of $\mathit{enc}$) $\emptyset[\trv \mapsto \tr],i \models_{\Pi(\ekripke)} \mathit{enc}(\trv,\LTLprevious \psi,v_\mathfrak{E})$.

\smallskip
\noindent\textbf{Case} $\varphi = \varphi_1 \LTLuntil \varphi_2$: The equivalence can be established as follows. $\ekripke,\tr,i,\mathfrak{E} \models \varphi_1 \LTLuntil \varphi_2$ iff (semantics of \yltlso) $\exists k \geq i$ such that $\ekripke,\tr,k,\mathfrak{E} \models \varphi_2$ and for all $i \leq j \leq k$ it holds that $\ekripke,\tr,j,\mathfrak{E} \models \varphi_1$ iff (induction hypothesis) $\exists k \geq i$ such that $\emptyset[\trv \mapsto \tr],k \models_{\Pi(\ekripke)} \mathit{enc}(\trv,\varphi_2,v_\mathfrak{E})$ and for all $i \leq j \leq k$ it holds that $\emptyset[\trv \mapsto \tr],j \models_{\Pi(\ekripke)} \mathit{enc}(\trv,\varphi_1,v_\mathfrak{E})$ iff (semantics of HyperQPTL) $\emptyset[\trv \mapsto \tr],i \models_{\Pi(\ekripke)} \mathit{enc}(\trv,\varphi_1,v_\mathfrak{E}) \LTLuntil \mathit{enc}(\trv,\varphi_2,v_\mathfrak{E})$ iff (definition of $\mathit{enc}$) $\emptyset[\trv \mapsto \tr],i \models_{\Pi(\ekripke)} \mathit{enc}(\trv,\varphi_1 \LTLuntil \varphi_2,v_\mathfrak{E})$.

\smallskip
\noindent\textbf{Case} $\varphi = \varphi_1 \U^- \varphi_2$: The equivalence can be established as follows. $\ekripke,\tr,i,\mathfrak{E} \models \varphi_1 \U^- \varphi_2$ iff (semantics of \yltlso) $\exists k \leq i$ such that $\ekripke,\tr,k,\mathfrak{E} \models \varphi_2$ and for all $i \geq j \geq k$ it holds that $\ekripke,\tr,j,\mathfrak{E} \models \varphi_1$ iff (induction hypothesis) $\exists k \leq i$ such that $\emptyset[\trv \mapsto \tr],k \models_{\Pi(\ekripke)} \mathit{enc}(\trv,\varphi_2,v_\mathfrak{E})$ and for all $i \geq j \geq k$ it holds that $\emptyset[\trv \mapsto \tr],j \models_{\Pi(\ekripke)} \mathit{enc}(\trv,\varphi_1,v_\mathfrak{E})$ iff (semantics of HyperQPTL) $\emptyset[\trv \mapsto \tr],i \models_{\Pi(\ekripke)} \mathit{enc}(\trv,\varphi_1,v_\mathfrak{E}) \U^- \mathit{enc}(\trv,\varphi_2,v_\mathfrak{E})$ iff (definition of $\mathit{enc}$) $\emptyset[\trv \mapsto \tr],i \models_{\Pi(\ekripke)} \mathit{enc}(\trv,\varphi_1 \U^- \varphi_2,v_\mathfrak{E})$.

\smallskip
\noindent\textbf{Case} $\varphi = \K_a \psi$: The equivalence can be established as follows. $\ekripke,\tr,i,\mathfrak{E} \models \K_a \psi$ iff (semantics of \yltlso) $\forall \tr' \in T : \tr[0,i] =_{\Omega(a)} \tr'[0,i] \rightarrow \ekripke,\tr',i,\mathfrak{E} \models \psi$ iff (induction hypothesis and renaming $\trv$ to $\rho$)  $\forall \tr' \in T : \tr[0,i] =_{\Omega(a)} \tr'[0,i] \rightarrow \emptyset[\rho \mapsto \tr'],i \models_{\Pi(\ekripke)} \mathit{enc}(\rho,\psi,v_\mathfrak{E})$ iff  $\forall \tr' \in T : (\emptyset[\trv \mapsto \tr, \rho \mapsto \tr'],i \models_{\Pi(\ekripke)} \LTLpastglobally (\bigwedge_{\ap \in \Omega(a)} \ap_\trv \leftrightarrow \ap_\rho)) \rightarrow \emptyset[\rho \mapsto \tr'],i \models_{\Pi(\ekripke)} \mathit{enc}(\rho,\psi,v_\mathfrak{E})$ iff (HyperQPTL semantics and $\trv$ not bound in consequent) $\emptyset[\trv \mapsto \tr],i \models_{\Pi(\ekripke)} \forall \rho . \LTLpastglobally (\bigwedge_{\ap \in \Omega(a)} \ap_\trv \leftrightarrow \ap_\rho) \rightarrow \mathit{enc}(\rho,\psi,v_\mathfrak{E})$ iff (definition of $\mathit{enc}$) $\emptyset[\trv \mapsto \tr],i \models_{\Pi(\ekripke)} \mathit{enc}(\trv,\K_a \psi,v_\mathfrak{E})$.

\smallskip
\noindent\textbf{Case} $\varphi = \exists X . \psi$: By assumption, we know that $\varphi$ is fully conjunctive and flat. This lets us use the lemmas that we will prove in the remainder of the section to conclude the following equivalence: $\ekripke,\tr,i,\mathfrak{E} \models \exists X . \psi$ \text{ iff } (\yltlso semantics) $\exists Y \subseteq (2^\AP)^\omega$ s.t. $\ekripke,\tr,i,\mathfrak{E}[X \mapsto Y] \models \psi$ iff (\Cref{lem:yltlsoequivtoright} for ``$\Rightarrow$'' and \Cref{lem:equivdirectionone} for ``$\Leftarrow$'') $\forall \tr' \in T : \emptyset[\trv \mapsto \tr, \rho \mapsto \tr'],i \models_{\Pi(\ekripke)} \mathit{enc}(\trv,\varphi,v_\mathfrak{E}[X \mapsto (\rho,\top)]) \lor \emptyset[\trv \mapsto \tr, \rho \mapsto \tr'],i \models_{\Pi(\ekripke)} \mathit{enc}(\trv,\varphi,v_\mathfrak{E}[X \mapsto (\rho,\bot)])$ iff (HyperQPTL semantics) $\emptyset[\trv \mapsto \tr],i \models_{\Pi(\ekripke)}  \forall \rho . \mathit{enc}(\trv,\psi,v_\mathfrak{E}[X \mapsto (\top)]) \lor \mathit{enc}(\trv,\psi,v_\mathfrak{E}[X \mapsto (\rho,\bot)])$. Note that we did not need to apply the induction hypothesis here, as we could use the two lemmas that prove a stronger statements for the subformula $\psi$.
\end{proof}

The remainder of this section introduces and proves the Lemmas~\ref{lem:equivdirectionone} and \ref{lem:yltlsoequivtoright}.  The proof of the former uses the following fact about fully conjunctive and second-order flat \yltlso formulas: They either allow exactly one assignment to a second-order variable, or allow any assignment to the variable. We establish this fact in the following lemma before moving on to the Lemmas~\ref{lem:equivdirectionone} and \ref{lem:yltlsoequivtoright} that prove the two directions of the equivalence between a fully conjunctive and second-order flat \yltlso formula and its encoding and were used in the proof of \Cref{thm:yltlsomc}.

\begin{lemma}\label{lem:uniquesosolutions}
    Let $\ekripke$ be an extended transition system, $\tr \in \Pi(\ekripke)$ a trace, $i \geq 0$ a time point, $\varphi$ a fully conjunctive and second-order flat \yltlso formula, and $\Theta$ a second-order assignment. For all $X \in \sov$ and $Y \in (2^\AP)^\omega$ such that $\ekripke,\tr,i,\Theta[X \mapsto Y] \models \varphi $:
    \begin{align*}
    &\forall \Theta' : \forall Z \in (2^\AP)^\omega : Y \neq Z \Rightarrow \ekripke,\tr,i,\Theta'[X \mapsto Z] \not\models \varphi \enspace \text{, or}\\
    &\forall Z \in (2^\AP)^\omega : \ekripke,\tr,i,\Theta[X \mapsto Z] \models \varphi \enspace .
    \end{align*}
\end{lemma}

\begin{proof}
Via structural induction over $\varphi$.\\
\textbf{Case} $\varphi = \ap$: Since in this case the satisfaction is independent of $\Theta$, we have $\ekripke,\tr,i,\Theta[X \mapsto Z] \models \ap$ for all $Z$ if there is some $Y$ such that $\ekripke,\tr,i,\Theta[X \mapsto Y] \models \ap$.

\smallskip
\noindent\textbf{Case} $\varphi = \smash{X' \xrsquigarrow{A} \psi}$ : Since $\varphi$ is second-order flat we can treat this as a base case and we have $\mathit{Cause}(\psi,\tr,i,A,T,\Theta) = \mathit{Cause}(\psi,\tr,i,A,T,\Theta')$ for all second-order assignments $\Theta'$. Now, consider $X \in \sov$ and  $Y \in (2^\AP)^\omega$ such that $\ekripke,\tr,i,\Theta[X \mapsto Y] \models \varphi$. If $X = X'$, we must have $Y = \mathit{Cause}(\psi,\tr,i,A,T,\Theta)$ and for all $\Theta'$ and $Z \neq Y$ we have $\ekripke,\tr,i,\Theta'[X \mapsto Z] \not\models \varphi$. For $X \neq X'$, we have $\Theta[X \mapsto Y](X') = \Theta[X \mapsto Z](X') = \mathit{Cause}(\psi,\tr,i,A,T,\Theta[X \mapsto Y]) = \mathit{Cause}(\psi,\tr,i,A,T,\Theta[X \mapsto Z])$ for all $Z$, hence also $\ekripke,\tr,i,\Theta[X \mapsto Z] \models \varphi$. 

\smallskip
\noindent\textbf{Case} $\varphi = \varphi_1 \land \varphi_2$: From \yltlso semantics, we can infer $\ekripke,\tr,i,\Theta[X \mapsto Y] \models \varphi_1 $ as well as $\ekripke,\tr,i,\Theta[X \mapsto Y] \models \varphi_2 $ and can apply the induction hypothesis. It is easy to see that if any of the two formulas $k \in \{1,2\}$ satisfies $\forall \Theta' \forall Z \in (2^\AP)^\omega : Y \neq Z \Rightarrow \ekripke,\tr,i,\Theta'[X \mapsto Z] \not\models \varphi_k$, then this also holds for the conjunction $\varphi_1 \land \varphi_2$ due to the semantics of \yltlso. Similarly, if both formulas are satisfied for all $Z \in (2^\AP)^\omega$ and the fixed assignment $\Theta$, then so is the conjunction $\varphi_1 \land \varphi_2$.

\smallskip
\noindent\textbf{Case} $\varphi = \LTLnext \psi$: It follows $\ekripke,\tr,i+1,\Theta[X \mapsto Y] \models \psi $ from the semantics of \yltlso and with the induction hypothesis
\begin{align*}
&\forall \Theta' : \forall Z \in (2^\AP)^\omega : Y \neq Z \Rightarrow \ekripke,\tr,i+1,\Theta[X \mapsto Z] \not\models \psi \enspace \text{, or}\\
&\forall Z \in (2^\AP)^\omega : \ekripke,\tr,i+1,\Theta[X \mapsto Z] \models \psi \enspace .
\end{align*}
Applying the semantics of $\LTLnext$, we infer
\begin{align*}
&\forall \Theta' : \forall Z \in (2^\AP)^\omega : Y \neq Z \Rightarrow \ekripke,\tr,i,\Theta[X \mapsto Z] \not\models \LTLnext \psi \enspace \text{, or}\\
&\forall Z \in (2^\AP)^\omega : \ekripke,\tr,i,\Theta[X \mapsto Z] \models \LTLnext \psi \enspace ,
\end{align*}
which closes this case.

\smallskip
\noindent\textbf{Case} $\varphi = \LTLprevious \psi$: This case works fully symmetrically to the previous case, as we know from the assumption that $i \geq 1$.

\smallskip
\noindent\textbf{Case} $\varphi = \LTLglobally \psi$: From \yltlso semantics, we can infer $\ekripke,\tr,j,\Theta[X \mapsto Y] \models \psi$ for all $j \geq i$ and apply the induction hypothesis. We make a case distinction: If there is at least one $j$ such that $\forall \Theta' : \forall Z \in (2^\AP)^\omega : Y \neq Z \Rightarrow \ekripke,\tr,j,\Theta'[X \mapsto Z] \not\models \psi$, then for these $\Theta'$ and $Z$ we also have $\ekripke,\tr,j,\Theta'[X \mapsto Z] \not\models \LTLglobally \psi$. If for all $j$ we have $\forall Z \in (2^\AP)^\omega : \ekripke,\tr,j,\Theta[X \mapsto Z] \models \psi$, then it follows for all such $Z$ that $\ekripke,\tr,i,\Theta[X \mapsto Z] \models \LTLglobally\psi$ by the semantics of $\LTLglobally$ .

\smallskip
\noindent\textbf{Case} $\varphi = \LTLpastglobally \psi$: This case works analogously to the previous case, except we pick $j \leq i$.

\smallskip
\noindent\textbf{Case} $\varphi = \K_a \psi$: From \yltlso semantics, we can infer $\ekripke,\tr_o,i,\Theta[X \mapsto Y] \models \psi$ for all observation-equivalent traces $\tr_o$ with $\tr[0,i] =_{\Omega(a)} \tr_o[0,i]$. We apply the induction hypothesis for $\psi$ and make a case distinction. If there is at least one such $\tr_o$ that satisfies $\forall \Theta' : \forall Z \in (2^\AP)^\omega : Y \neq Z \Rightarrow \ekripke,\tr_o,i,\Theta'[X \mapsto Z] \not\models \psi$, then we can conclude for all such $\Theta'$ and $Z$ that $\ekripke,\tr,i,\Theta'[X \mapsto Z] \not\models \K_a \psi$ due to the semantics of the epistemic operator. If there is no such $\tr_o$ we have $\forall Z \in (2^\AP)^\omega : \ekripke,\tr_o,i,\Theta[X \mapsto Z] \models \psi $ for all $\tr_o$. This lets us conclude $\ekripke,\tr,i,\Theta[X \mapsto Z] \models \K_a \psi$ for all these $Z$, which closes this case.

\smallskip
\noindent\textbf{Case} $\varphi = \exists X'. \psi$: From \yltlso semantics, we can infer that there is some $Y' \subseteq (2^\AP)^\omega$ such that $\ekripke,\tr,i,\Theta[X \mapsto Y,X' \mapsto Y'] \models \psi$. This lets us apply the induction hypothesis to infer
\begin{align*}
&\forall \Theta' : \forall Z \in (2^\AP)^\omega : Y \neq Z \Rightarrow \ekripke,\tr,i,\Theta'[X \mapsto Z] \not\models \psi \enspace \text{, or}\\
&\forall Z \in (2^\AP)^\omega : \ekripke,\tr,i,\Theta[X \mapsto Z,X' \mapsto Y'] \models \psi \enspace ,
\end{align*}
Since the first case applies to any second-order assignment, we cannot find an instantiation of $X'$ that satisfies $\psi$. In the second case, we have found the satisfying assignment with $Y'$. This allows us to infer the following, using the semantics of $\exists$.
\begin{align*}
&\forall \Theta' : \forall Z \in (2^\AP)^\omega : Y \neq Z \Rightarrow \ekripke,\tr,i,\Theta'[X \mapsto Z] \not\models \exists X'. \psi \enspace \text{, or}\\
&\forall Z \in (2^\AP)^\omega : \ekripke,\tr,i,\Theta'[X \mapsto Z] \models \exists X'. \psi \enspace ,
\end{align*}
\end{proof}

Now, we prove that the language of the HyperQPTL encoding over trace variables from the second-order mapping $v(X).1$ correspond to a solution for $X$, if the encoding is satisfied by the current trace and time point. This established that satisfaction of an encoding implies a satisfying assignment to the second-order variables of the encoded formula $\varphi$. More precisely, we are interested in the following set of traces, for some HyperQPTL formula $\varphi$:
\begin{align*}
    \Lang_{\traceassign,i,\rho,\Pi(\ekripke)}(\varphi) = \{ \tr \in (2^\AP)^\omega \mid \traceassign[\rho \mapsto \tr],i \models_{\Pi(\ekripke)} \varphi \} \enspace .
\end{align*}
Hence, this is the set of all trace that can be assigned to the trace variable $\rho$ in the given trace assignment $\traceassign$ such that the assignment at time point $i$ interpreted in the set of traces $\Pi(\ekripke)$ satisfies the HyperQPTL $\varphi.$

\begin{lemma}\label{lem:equivdirectionone}
Let $\ekripke$ be an extended transition system, $\traceassign$ be a trace assignment, $\tr \in \Pi(\kripke)$ a trace, $i \geq 0$ a time point, $\varphi$ a fully conjunctive and second-order flat \yltlso formula and $v$ a non-empty mapping from second-order variables to $\TV \times \{\top,\bot \}$. If the following statement holds:
    \begin{align*}
    \forall X \in \mathfrak{D}(v) : \forall \tr' \in (2^\AP)^\omega : \ &\traceassign[\trv \mapsto \tr, v(X).1 \mapsto \tr'],i \models_{\Pi(\ekripke)} \mathit{enc}(\trv,\varphi,v[ X \mapsto (v(X).1,\top)]) \ \lor\\ &\traceassign[\trv \mapsto \tr, v(X).1 \mapsto \tr'],i \models_{\Pi(\ekripke)} \mathit{enc}(\trv,\varphi,v[ X \mapsto (v(X).1,\bot)]) \enspace ,
    \end{align*}
then also the following statement holds
    \begin{align*}
    &\ekripke,\tr,i,\mathfrak{E}[\forall X \in \mathfrak{D}(v): X \mapsto \Lang_{\traceassign[\trv \mapsto \tr],i,v(X).1,\Pi(\ekripke)}(\mathit{enc}(\trv,\varphi,v[ X \mapsto (v(X).1,\top)]))] \models \varphi \enspace \text{, and }\\
    &\ekripke,\tr,i,\mathfrak{E}[\forall X \in \mathfrak{D}(v): X \mapsto \overline{\Lang_{\traceassign[\trv \mapsto \tr],i,v(X).1,\Pi(\ekripke)}(\mathit{enc}(\trv,\varphi,v[ X \mapsto (v(X).1,\bot)]))}] \models \varphi \enspace .
    \end{align*}
\end{lemma}

\begin{proof}
Via structural induction over $\varphi$.\\
\textbf{Case} $\varphi = \ap$: This is trivial, as $\mathit{enc}(\trv,\ap,v[ X \mapsto (v(X).1,\top)]) = \ap_\trv = \mathit{enc}(\trv,\ap,v[ X \mapsto (v(X).1,\bot)])$, and if  $\traceassign[\trv \mapsto \tr, v(X).1 \mapsto \tr'],i \models_{\Pi(\ekripke)} \ap_\trv$ then $\ekripke, \tr,i,\Theta \models \ap$ for any assignment $\Theta$.

\smallskip
\noindent\textbf{Case} $\varphi = \smash{X \xrsquigarrow{A} \psi}$ : Since our formula is second-order flat and hence $\psi$ must be an LTL formula, we can treat this like a base case. From the definition of the encoding function $\mathit{enc}$ and the definition the set $\mathit{Cause}$, we can infer that
\begin{align*}
&\Lang_{\traceassign[\trv \mapsto \tr],i,v(X).1,\Pi(\ekripke)}(\mathit{enc}(\trv, \varphi,v[ X \mapsto (v(X).1,\top)])) \\ &= \mathit{Cause}(\psi,\tr,i,A,T,\Theta) =\\
&\overline{\Lang_{\traceassign[\trv \mapsto \tr],i,v(X).1,\Pi(\ekripke)}(\mathit{enc}(\trv, \varphi,v[ X \mapsto (v(X).1,\bot)]))} \enspace ,
\end{align*}
and from this the consequent follows immediately from the semantics of \yltlso.

\smallskip
\noindent\textbf{Case} $\varphi = \varphi_1 \land \varphi_2$: From the definition of $\mathit{enc}$, we can infer $\forall X \in \mathfrak{D}(v) : \forall t' \in (2^\AP)^\omega : $
\begin{align*}
&\big( \traceassign[\trv \mapsto \tr, v(X).1 \mapsto \tr'],i \models_{\Pi(\ekripke)} \mathit{enc}(\trv,\varphi_1,v[ X \mapsto (v(X).1,\top)]) \ \land \\ &\traceassign[\trv \mapsto \tr, v(X).1 \mapsto \tr'],i \models_{\Pi(\ekripke)} \mathit{enc}(\trv,\varphi_2,v[ X \mapsto (v(X).1,\top)]) \big) \ \lor\\  &\big(\traceassign[\trv \mapsto \tr, v(X).1 \mapsto \tr'],i \models_{\Pi(\ekripke)} \mathit{enc}(\trv,\varphi_1,v[ X \mapsto (v(X).1,\bot)]) \ \land \\ &\traceassign[\trv \mapsto \tr, v(X).1 \mapsto \tr'],i \models_{\Pi(\ekripke)} \mathit{enc}(\trv,\varphi_2,v[ X \mapsto (v(X).1,\bot)]) \enspace .
\end{align*}
This lets us apply the induction hypothesis for $\varphi_1$ and $\varphi_2$ to conclude for $k \in \{1,2\}$:
\begin{align*}
    &\ekripke,\tr,i,\mathfrak{E}[\forall X \in \mathfrak{D}(v): X \mapsto \Lang_{\traceassign[\trv \mapsto \tr],i,v(X).1,\Pi(\ekripke)}(\mathit{enc}(\trv,\varphi_k,v[ X \mapsto (v(X).1,\top)]))] \models \varphi_k \enspace ,\\
    &\ekripke,\tr,i,\mathfrak{E}[\forall X \in \mathfrak{D}(v): X \mapsto \overline{\Lang_{\traceassign[\trv \mapsto \tr],i,v(X).1,\Pi(\ekripke)}(\mathit{enc}(\trv,\varphi_k,v[ X \mapsto (v(X).1,\bot)]))}] \models \varphi_k \enspace .
\end{align*}
So far, we do not know whether the conjuncts can be satisfied by the same second-order assignment. To close this case, we need to establish that these satisfaction relations also hold when each $X$ is assigned to the $v(X).1$-language of the encoding of $\varphi_1 \land \varphi_2$, as the rest then follows from the semantics of \yltlso. If the $v(X).1$-language of the positive encoding and of the complement of the negative encoding are different, we know with \Cref{lem:uniquesosolutions} that the satisfaction holds for any set assigned to $X$. If not, i.e., we have for a $k \in \{1,2\}$:
\begin{align*}
    &\Lang_{\traceassign[\trv \mapsto \tr],i,v(X).1,\Pi(\ekripke)}(\mathit{enc}(\trv,\varphi_k,v[ X \mapsto (v(X).1,\top)])) =\\
    &\overline{\Lang_{\traceassign[\trv \mapsto \tr],i,v(X).1,\Pi(\ekripke)}(\mathit{enc}(\trv,\varphi_k,v[ X \mapsto (v(X).1,\bot)]))} \enspace ,
\end{align*}
then this means all traces $\tr \in (2^\AP)^\omega$ are exclusively either in the positive encoding or in the negative encoding. Then we can conclude with out initial assumption, which said that any trace must be in the positive or the negative encoding of the conjunction $\varphi_1 \land \varphi_2$ that 
\begin{align*}
    &\Lang_{\traceassign[\trv \mapsto \tr],i,v(X).1,\Pi(\ekripke)}(\mathit{enc}(\trv,\varphi_k,v[ X \mapsto (v(X).1,*)])) =\\
    &\Lang_{\traceassign[\trv \mapsto \tr],i,v(X).1,\Pi(\ekripke)}(\mathit{enc}(\trv,\varphi_1 \land \varphi_2,v[ X \mapsto (v(X).1,*)])) \enspace ,
\end{align*}
for $* \in \{\top,\bot\}$. This is because the conjunction $\varphi_1 \land \varphi_2$ builds the intersection of the two $v(X).1$-languages, and hence can only make the $v(X).1$-language of $\varphi_k$ smaller -- irrespective of the polarity of the encoding $*$. However, since we inferred that the positive and negative encodings of $\varphi_k$ perfectly partition the set of all traces, and assumed that all traces are in either the positive or the negative encoding of the conjunction at the start of this case, they cannot really make the $v(X).1$-languages smaller at all. From this equality and the fact that any assignment works for the earlier mentioned case, we can then infer
\begin{align*}
    &\ekripke,\tr,i,\mathfrak{E}[\forall X \in \mathfrak{D}(v): X \mapsto \Lang_{\traceassign[\trv \mapsto \tr],i,v(X).1,\Pi(\ekripke)}(\mathit{enc}(\trv,\varphi_1 \land \varphi_2,v[ X \mapsto (v(X).1,\top)]))] \models \varphi_k  \enspace ,\\
    &\ekripke,\tr,i,\mathfrak{E}[\forall X \in \mathfrak{D}(v): X \mapsto \overline{\Lang_{\traceassign[\trv \mapsto \tr],i,v(X).1,\Pi(\ekripke)}(\mathit{enc}(\trv,\varphi_1 \land \varphi_2,v[ X \mapsto (v(X).1,\bot)]))}] \models \varphi_k \enspace ,
\end{align*}
for both $k$, which closes this case due to the semantics of \yltlso.

\smallskip
\noindent\textbf{Case} $\varphi = \LTLnext \psi$: This case is a straightforward stepping through the semantics, i.e., from the definition of $\mathit{enc}$ and the semantics of HyperQPTL, we can infer $\forall X \in \mathfrak{D}(v) : \forall t' \in (2^\AP)^\omega : $
\begin{align*}
&\traceassign[\trv \mapsto \tr, v(X).1 \mapsto \tr'],i+1 \models_{\Pi(\ekripke)} \mathit{enc}(\trv,\varphi_1,v[ X \mapsto (v(X).1,\top)]) \ \lor\\  & \traceassign[\trv \mapsto \tr, v(X).1 \mapsto \tr'],i+1 \models_{\Pi(\ekripke)} \mathit{enc}(\trv,\varphi_1,v[ X \mapsto (v(X).1,\bot)]) \enspace .
\end{align*}
This lets us apply the induction hypothesis for $\psi$ with $i+1$ to infer
\begin{align*}
    &\ekripke,\tr,i+1,\mathfrak{E}[\forall X \in \mathfrak{D}(v): X \mapsto \Lang_{\traceassign[\trv \mapsto \tr],i+1,v(X).1,\Pi(\ekripke)}(\mathit{enc}(\trv,\psi,v[ X \mapsto (v(X).1,\top)]))] \models \psi \enspace ,\\
    &\ekripke,\tr,i+1,\mathfrak{E}[\forall X \in \mathfrak{D}(v): X \mapsto \overline{\Lang_{\traceassign[\trv \mapsto \tr],i+1,v(X).1,\Pi(\ekripke)}(\mathit{enc}(\trv,\psi,v[ X \mapsto (v(X).1,\bot)]))}] \models \psi \enspace .
\end{align*}
This can be rewritten with the definition of $\Lang$ and $\mathit{enc}$, as well as the semantics of \yltlso and HyperQPTL to:
\begin{align*}
    &\ekripke,\tr,i,\mathfrak{E}[\forall X \in \mathfrak{D}(v): X \mapsto \Lang_{\traceassign[\trv \mapsto \tr],i,v(X).1,\Pi(\ekripke)}(\mathit{enc}(\trv,\LTLnext  \psi,v[ X \mapsto (v(X).1,\top)]))] \models \LTLnext \psi \enspace ,\\
    &\ekripke,\tr,i,\mathfrak{E}[\forall X \in \mathfrak{D}(v): X \mapsto \overline{\Lang_{\traceassign[\trv \mapsto \tr],i,v(X).1,\Pi(\ekripke)}(\mathit{enc}(\trv,\LTLnext \psi,v[ X \mapsto (v(X).1,\bot)]))}] \models \LTLnext \psi \enspace ,
\end{align*}
which closes this case.

\smallskip
\noindent\textbf{Case} $\varphi = \LTLprevious \psi$: Works analogously to the previous case, as we can infer from the assumption that $i > 0$.

\smallskip
\noindent\textbf{Case} $\varphi = \LTLglobally \psi$: This case works with the same principles as used for the conjunction $\varphi_1 \land \varphi_2$. From the definition of $\mathit{enc}$ and the semantics of $\LTLglobally$, we can infer $\forall X \in \mathfrak{D}(v) : \forall \tr' \in (2^\AP)^\omega : $
\begin{align*}
&\big(\forall j \geq i : \traceassign[\trv \mapsto \tr, v(X).1 \mapsto \tr'],j \models_{\Pi(\ekripke)} \mathit{enc}(\trv,\psi,v[ X \mapsto (v(X).1,\top)]) \ \big)\lor\\  &\big(\forall j \geq i : \traceassign[\trv \mapsto \tr, v(X).1 \mapsto \tr'],j \models_{\Pi(\ekripke)} \mathit{enc}(\trv,\psi,v[ X \mapsto (v(X).1,\bot)]) \big)\enspace .
\end{align*}
This lets us apply the induction hypothesis for $\psi$ and all $j$ to conclude:
\begin{align*}
    &\ekripke,\tr,j,\mathfrak{E}[\forall X \in \mathfrak{D}(v): X \mapsto \Lang_{\traceassign[\trv \mapsto \tr],i,v(X).1,\Pi(\ekripke)}(\mathit{enc}(\trv,\psi,v[ X \mapsto (v(X).1,\top)]))] \models \psi \enspace ,\\
    &\ekripke,\tr,j,\mathfrak{E}[\forall X \in \mathfrak{D}(v): X \mapsto \overline{\Lang_{\traceassign[\trv \mapsto \tr],i,v(X).1,\Pi(\ekripke)}(\mathit{enc}(\trv,\psi,v[ X \mapsto (v(X).1,\bot)]))}] \models \psi \enspace .
\end{align*}
Again, we have not inferred that $\psi$ is satisfied by the same second-order assignment for all $j \geq i$, so we cannot immediately conclude that $\LTLglobally \psi$ holds at $i$. Like for the conjunction, we can again infer with \Cref{lem:uniquesosolutions} that if the $v(X).1$-language of the positive encoding and of the complement of the negative encoding are different for some $j$, that the satisfaction of $\psi$ at $j$ holds for any set assigned to $X$. Let $j$ be now a time point where this does not hold, i.e., we have:
\begin{align*}
    &\Lang_{\traceassign[\trv \mapsto \tr],j,v(X).1,\Pi(\ekripke)}(\mathit{enc}(\trv,\psi,v[ X \mapsto (v(X).1,\top)])) =\\
    &\overline{\Lang_{\traceassign[\trv \mapsto \tr],j,v(X).1,\Pi(\ekripke)}(\mathit{enc}(\trv,\psi,v[ X \mapsto (v(X).1,\bot)]))} \enspace ,
\end{align*}
We can again conclude with the observation that the universal condition of $\LTLglobally$ can only make the respective sets smaller, while they perfectly partition the set of traces before and still have to cover all of the traces after applying the universal condition of $\LTLglobally$, that
\begin{align*}
    &\Lang_{\traceassign[\trv \mapsto \tr],j,v(X).1,\Pi(\ekripke)}(\mathit{enc}(\trv,\psi,v[ X \mapsto (v(X).1,*)])) =\\
    &\Lang_{\traceassign[\trv \mapsto \tr],i,v(X).1,\Pi(\ekripke)}(\mathit{enc}(\trv,\LTLglobally \psi,v[ X \mapsto (v(X).1,*)])) \enspace ,
\end{align*}
for all such $j$ and $* \in \{\top,\bot\}$. With the observation that all assignments to $X$ work at the positions $j$ that do not fall under this, we can then conclude
\begin{align*}
    &\ekripke,\tr,j,\mathfrak{E}[\forall X \in \mathfrak{D}(v): X \mapsto \Lang_{\traceassign[\trv \mapsto \tr],i,v(X).1,\Pi(\ekripke)}(\mathit{enc}(\trv,\LTLglobally \psi,v[ X \mapsto (v(X).1,\top)]))] \models \psi  \enspace ,\\
    &\ekripke,\tr,j,\mathfrak{E}[\forall X \in \mathfrak{D}(v): X \mapsto \overline{\Lang_{\traceassign[\trv \mapsto \tr],i,v(X).1,\Pi(\ekripke)}(\mathit{enc}(\trv,\LTLglobally \psi,v[ X \mapsto (v(X).1,\bot)]))}] \models \psi \enspace ,
\end{align*}
for all $j \geq i$, which concludes this case after applying the semantics of $\LTLglobally$ .

\smallskip
\noindent\textbf{Case} $\varphi = \LTLpastglobally \psi$: Like for the $\LTLprevious$ operator, this case works full analogously to the case of $\varphi = \LTLglobally \psi$, except $j$ is chosen as $j \leq i$.

\smallskip
\noindent\textbf{Case} $\varphi = \K_a \psi$: The epistemic operator states a universal condition like the $\LTLglobally$ operator, except the quantification is over observation-equivalent traces instead of time points. The proof steps follow the same structure. From the definition of $\mathit{enc}$ and the semantics of $\K_a$, we can infer $\forall X \in \mathfrak{D}(v) : \forall t' \in (2^\AP)^\omega : $
\begin{align*}
&\big(\forall \tr_o \in T^a_\tr : \traceassign[\trv \mapsto \tr_o, v(X).1 \mapsto \tr'],j \models_{\Pi(\ekripke)} \mathit{enc}(\trv,\psi,v[ X \mapsto (v(X).1,\top)]) \ \big) \ \lor\\  &\big(\forall \tr_o \in T^a_\tr : \traceassign[\trv \mapsto \tr_o, v(X).1 \mapsto \tr'],j \models_{\Pi(\ekripke)} \mathit{enc}(\trv,\psi,v[ X \mapsto (v(X).1,\bot)]) \big)\enspace ,
\end{align*}
where $T^a_\tr = \{ \tr_o \mid \tr[0,i] =_{\Omega(a)} \tr_o[0,i]\}$ is the set of traces that are observation-equivalent to $\tr$ up to $i$. We apply the induction hypothesis for $\psi$ and all $t_o$ to conclude:
\begin{align*}
    &\ekripke,\tr_o,j,\mathfrak{E}[\forall X \in \mathfrak{D}(v): X \mapsto \Lang_{\traceassign[\trv \mapsto \tr_o],i,v(X).1,\Pi(\ekripke)}(\mathit{enc}(\trv,\psi,v[ X \mapsto (v(X).1,\top)]))] \models \psi \enspace ,\\
    &\ekripke,\tr_o,j,\mathfrak{E}[\forall X \in \mathfrak{D}(v): X \mapsto \overline{\Lang_{\traceassign[\trv \mapsto \tr_o],i,v(X).1,\Pi(\ekripke)}(\mathit{enc}(\trv,\psi,v[ X \mapsto (v(X).1,\bot)]))}] \models \psi \enspace .
\end{align*}
With \Cref{lem:uniquesosolutions} we can again restrict the discussion to $\tr_o$ where:
\begin{align*}
    &\Lang_{\traceassign[\trv \mapsto \tr_o],j,v(X).1,\Pi(\ekripke)}(\mathit{enc}(\trv,\psi,v[ X \mapsto (v(X).1,\top)])) =\\
    &\overline{\Lang_{\traceassign[\trv \mapsto \tr_o],j,v(X).1,\Pi(\ekripke)}(\mathit{enc}(\trv,\psi,v[ X \mapsto (v(X).1,\bot)]))} \enspace ,
\end{align*}
as any assignment of $X$ works for the others traces. The universal condition of $\K_a$ again only makes the respective sets smaller and they perfectly partition the set of traces before and cover all  traces after applying $\K_a$, such that
\begin{align*}
    &\Lang_{\traceassign[\trv \mapsto \tr_o],i,v(X).1,\Pi(\ekripke)}(\mathit{enc}(\trv,\psi,v[ X \mapsto (v(X).1,*)])) =\\
    &\Lang_{\traceassign[\trv \mapsto \tr],i,v(X).1,\Pi(\ekripke)}(\mathit{enc}(\trv,\K_a \psi,v[ X \mapsto (v(X).1,*)])) \enspace ,
\end{align*}
for all such $\tr_o$ and $* \in \{\top,\bot\}$. All assignments to $X$ work at the other traces $t_o$ where the induction hypothesis produced two solutions, so we can conclude
\begin{align*}
    &\ekripke,\tr_o,i,\mathfrak{E}[\forall X \in \mathfrak{D}(v): X \mapsto \Lang_{\traceassign[\trv \mapsto \tr],i,v(X).1,\Pi(\ekripke)}(\mathit{enc}(\trv,\K_a \psi,v[ X \mapsto (v(X).1,\top)]))] \models \psi  \enspace ,\\
    &\ekripke,\tr_o,i,\mathfrak{E}[\forall X \in \mathfrak{D}(v): X \mapsto \overline{\Lang_{\traceassign[\trv \mapsto \tr],i,v(X).1,\Pi(\ekripke)}(\mathit{enc}(\trv,\K_a \psi,v[ X \mapsto (v(X).1,\bot)]))}] \models \psi \enspace ,
\end{align*}
for all $\tr_o \in T^a_\tr$, which lets us close this case through applying the semantics of $\K_a$.

\smallskip
\noindent\textbf{Case} $\varphi = \exists Y .\ \psi$: From the definition of $\mathit{enc}$, we can infer $\forall X \in \mathfrak{D}(v) : \forall \tr',\tr'' \in (2^\act)^\omega :$
\begin{align*}
&\traceassign[\trv \mapsto \tr, v(X).1 \mapsto \tr',\rho \mapsto \tr''],i \models_{\Pi(\ekripke)} \mathit{enc}(\trv,\psi,v[ X \mapsto (v(X).1,\top), Y \mapsto (v(X).1,\top)]) \  \lor\\  & \traceassign[\trv \mapsto \tr, v(X).1 \mapsto \tr',\rho \mapsto \tr''],i \models_{\Pi(\ekripke)} \mathit{enc}(\trv,\psi,v[ X \mapsto (v(X).1,\top), Y \mapsto (v(X).1,\bot)]) \ \lor\\  &\traceassign[\trv \mapsto \tr, v(X).1 \mapsto \tr',\rho \mapsto \tr''],i \models_{\Pi(\ekripke)} \mathit{enc}(\trv,\psi,v[ X \mapsto (v(X).1,\bot), Y \mapsto (v(X).1,\top)]) \  \lor\\  & \traceassign[\trv \mapsto \tr, v(X).1 \mapsto \tr',\rho \mapsto \tr''],i \models_{\Pi(\ekripke)} \mathit{enc}(\trv,\psi,v[ X \mapsto (v(X).1,\bot), Y \mapsto (v(X).1,\bot)]) \ \enspace .
\end{align*}
Consider the mapping $v' = v[Y \mapsto (\rho,\top)]$. Since it is guaranteed that at least one of the four cases is satisfied for any $\tr',\tr'' \in (2^\act)^\omega$, we can conclude $\forall X \in \mathfrak{D}(v') : \forall \tr' \in (2^\act)^\omega :$
\begin{align*}
     &\traceassign[\trv \mapsto \tr, v'(X).1 \mapsto \tr'],i \models_{\Pi(\ekripke)} \mathit{enc}(\trv,\psi,v'[ X \mapsto (v(X).1,\top)]) \ \lor\\  &\traceassign[\trv \mapsto \tr, v'(X).1 \mapsto \tr'],i \models_{\Pi(\ekripke)} \mathit{enc}(\trv,\psi,v'[ X \mapsto (v'(X).1,\bot)]) \enspace ,
\end{align*}
This lets us immediately apply the induction hypothesis to conclude:
\begin{align*}
    &\ekripke,\tr,i,\mathfrak{E}[\forall X \in \mathfrak{D}(v'): X \mapsto \Lang_{\traceassign[\trv \mapsto \tr],i,v(X).1,\Pi(\ekripke)}(\mathit{enc}(\trv,\psi,v[ X \mapsto (v(X).1,\top)]))] \models \psi \enspace ,\\
    &\ekripke,\tr,i,\mathfrak{E}[\forall X \in \mathfrak{D}(v'): X \mapsto \overline{\Lang_{\traceassign[\trv \mapsto \tr],i,v(X).1,\Pi(\ekripke)}(\mathit{enc}(\trv,\psi,v[ X \mapsto (v(X).1,\bot)]))}] \models \psi \enspace .
\end{align*}
Since $Y \in \mathfrak{D}(v')$, this means we have found a satisfying assignment to $Y$, which closes this case and the overall structural induction.
\end{proof}

\Cref{lem:equivdirectionone} shows one direction of the equivalence between the encoding and a fully conjunctive and flat \yltlso formula encoded in the proof of \Cref{thm:yltlsomc}, showing that the languages built over the trace variables of a satisfied encoding correspond to a satisfying assignment of the second-order variables. The following lemma establishes the reverse direction, stating that if a \yltlso formula is satisfied by a current trace, time point and second-order assignment, then all traces in a set assigned to a second-order variable $X$ satisfy the $X$-positive encoding of $\varphi$, while all traces not in the set satisfy $X$-negative encoding of $\varphi$. This implies the fact that if a current trace, time point and second-order assignment satisfy $\exists X .\psi$, where psi is fully conjunctive and flat, the corresponding trace assignment and time point satisfy the encoding of $\exists X .\psi$, which is used to close this case in the proof of \Cref{thm:yltlsomc}.

\begin{lemma}\label{lem:yltlsoequivtoright}
Let $\ekripke$ be an extended transition system, $\tr \in \Pi(\ekripke)$ a trace, $i \geq 0$ a time point,  $\Theta$ a second-order assignment, and $\varphi$ a \yltlso formula that is fully conjunctive and second-order flat. If $\varphi$ is satisfied by these semantic parameters, i.e., if $\ekripke,\tr,i,\Theta \models \varphi$, then we can infer for a trace assignment $\traceassign$ and variable mapping $v$ such that $v(X) = (\rho,\top)$ iff $\traceassign(\rho) \in \Theta(X)$ and $v(X) = (\rho,\bot)$ iff $\traceassign(\rho) \notin \Theta(X)$ for all $X \in \mathfrak{D}(\Theta)$ and some $\rho$, that for all  $X' \in  \mathfrak{D}(\Theta)$:
\begin{align*}
    &\traceassign[\trv \mapsto \tr, v(X').1 \mapsto \tr'],i \models_{\Pi(\ekripke)} \mathit{enc}(\trv,\varphi,v[X' \mapsto (v(X').1,\top)]) \text{ for all } \tr' \in \Theta(X') \text{, and}\\
    &\traceassign[\trv \mapsto \tr, v(X').1 \mapsto \tr'],i \models_{\Pi(\ekripke)} \mathit{enc}(\trv,\varphi,v[X' \mapsto (v(X').1,\bot)]) \text{ for all } \tr' \notin \Theta(X') \enspace .
\end{align*}
\end{lemma}

\begin{proof}
Via structural induction over $\varphi$.\\
\textbf{Case} $\varphi = \ap$: The encoding $\mathit{enc}(\trv,\ap,v') = \ap_\trv$ is the same for all variable mappings $v'$, and it is satisfied by all $\traceassign[\trv \mapsto \tr]$ since $\tr$ satisfies $\ap$ by assumption. This closes the case.

\smallskip
\noindent\textbf{Case} $\varphi = \smash{X'' \xrsquigarrow{A} \psi}$ : Since the formula is second-order flat we can treat this as a base case. By assumption and because $\psi$ is in LTL and hence its satisfaction only dependent on $\tr$, we know for all $\tr'' \in \Theta(X'')$ that all $\tr''' \leq^\mathit{A}_\tr \tr''$ satisfy $\ekripke,\tr'',i,\Theta' \models \psi$ and for all $\tr'' \notin \Theta(X'')$ there exists a $\tr''' \leq^\mathit{A}_\tr \tr''$ such that $\tr'',i,\Theta' \not\models_T \psi$ for all assignments $\Theta'$. With the definition of $\mathit{enc}$, this lets us conclude 
\begin{align*}
    &\traceassign'[\trv \mapsto \tr, v'(X'').1 \mapsto \tr''],i \models_{\Pi(\ekripke)} \mathit{enc}(\trv,\varphi,v'[X'' \mapsto (v'(X'').1,\top)]) \text{ for all } \tr'' \in \Theta(X'') \text{, and}\\
    &\traceassign'[\trv \mapsto \tr, v'(X'').1 \mapsto \tr''],i \models_{\Pi(\ekripke)} \mathit{enc}(\trv,\varphi,v'[X'' \mapsto (v'(X'').1,\bot)]) \text{ for all } \tr'' \notin \Theta(X'') \enspace ,
\end{align*}
for any trace assignment $\traceassign'$ and variable mapping $v'$. Since the initial assignment $\traceassign$ and mapping $v$ are constructed to fall under either of the two cases for $X''$, this in particular lets us infer the claim for any other second-order variable $Y \neq X''$ and the resulting assignments $\traceassign' = \traceassign[v(Y).1 \mapsto \tr']$ and $v[Y \mapsto (v(Y).1,\top)]$ with $\tr' \in \Theta(Y)$, as well as the other case with  $\traceassign' = \traceassign[v(Y).1 \mapsto \tr']$ and $v[Y \mapsto (v(Y).1,\bot)]$ with $\tr' \notin \Theta(Y)$.

\smallskip
\noindent\textbf{Case} $\varphi = \varphi_1 \land \varphi_2$: By the semantics of \yltlso we know the subformulas are satisfied by the same assignment: $\ekripke,\tr,i,\Theta \models \varphi_1$ and $\ekripke,\tr,i,\Theta \models \varphi_2$. We apply the induction hypothesis to infer for all trace assignments $\traceassign$ and variable mappings $v$ such that $v(X) = (\rho,\top)$ iff $\traceassign(\rho) \in \Theta(X)$ and $v(X) = (\rho,\bot)$ iff $\traceassign(\rho) \notin \Theta(X)$ for all $X \in \mathfrak{D}(\Theta)$ and some $\rho$, that for all  $X' \in  \mathfrak{D}(\Theta)$ and $k \in \{1,2\}$:
\begin{align*}
    &\traceassign[\trv \mapsto \tr, v(X').1 \mapsto \tr'],i \models_{\Pi(\ekripke)} \mathit{enc}(\trv,\varphi_k,v[X' \mapsto (v(X').1,\top)]) \text{ for all } \tr' \in \Theta(X') \text{, and}\\
    &\traceassign[\trv \mapsto \tr, v(X').1 \mapsto \tr'],i \models_{\Pi(\ekripke)} \mathit{enc}(\trv,\varphi_k,v[X' \mapsto (v(X').1,\bot)]) \text{ for all } \tr' \notin \Theta(X') \enspace .
\end{align*}
The we can simply use the semantics of HyperQPTL and the definition of the encoding function $\mathit{enc}$ to infer the claim for this case.

\smallskip
\noindent\textbf{Case} $\varphi = \LTLnext \psi$: From the semantics of \yltlso: $\ekripke,\tr,i+1,\Theta \models \psi$. With the induction hypothesis, we infer 
\begin{align*}
    &\traceassign[\trv \mapsto \tr, v(X').1 \mapsto \tr'],i+1 \models_{\Pi(\ekripke)} \mathit{enc}(\trv,\psi,v[X' \mapsto (v(X').1,\top)]) \text{ for all } \tr' \in \Theta(X') \text{, and}\\
    &\traceassign[\trv \mapsto \tr, v(X').1 \mapsto \tr'],i+1 \models_{\Pi(\ekripke)} \mathit{enc}(\trv,\psi,v[X' \mapsto (v(X').1,\bot)]) \text{ for all } \tr' \notin \Theta(X') \enspace .
\end{align*}
With the semantics of HyperQPTL and the definition of $\mathit{enc}$, the claim for this case follows.

\smallskip
\noindent\textbf{Case} $\varphi = \LTLprevious \psi$: Works analogously to the previous case, as we know $i \geq 1$ from the assumption.

\smallskip
\noindent\textbf{Case} $\varphi = \LTLglobally \psi$: From the semantics of \yltlso we know for all $j \geq i$: $\ekripke,\tr,j,\Theta \models \psi$. We apply the induction hypothesis for $\psi$ and all $j$ and infer 
\begin{align*}
    &\traceassign[\trv \mapsto \tr, v(X').1 \mapsto \tr'],j \models_{\Pi(\ekripke)} \mathit{enc}(\trv,\psi,v[X' \mapsto (v(X').1,\top)]) \text{ for all } \tr' \in \Theta(X') \text{, and}\\
    &\traceassign[\trv \mapsto \tr, v(X').1 \mapsto \tr'],j \models_{\Pi(\ekripke)} \mathit{enc}(\trv,\psi,v[X' \mapsto (v(X').1,\bot)]) \text{ for all } \tr' \notin \Theta(X') \enspace ,
\end{align*}
for all $j \geq i$. The claim for this case again follows form applying the semantics of HyperQPTL and the definition of $\mathit{enc}$.

\smallskip
\noindent\textbf{Case} $\varphi = \LTLpastglobally \psi$: The past version of the temporal operator again can be shown analogously to the future version, except we choose $j \leq i$.

\smallskip
\noindent\textbf{Case} $\varphi = \K_a \psi$: With the semantics of \yltlso, we know $\ekripke,\tr_o,i,\Theta \models \psi$ for all observation-equivalent traces $\tr_o$ with $\tr[0,i] =_{\Omega(a)} \tr_o[0,i]$. We apply the induction hypothesis for $\psi$ and each $\tr_o$ to infer:
\begin{align*}
    &\traceassign[\trv \mapsto \tr_o, v(X').1 \mapsto \tr'],i \models_{\Pi(\ekripke)} \mathit{enc}(\trv,\psi,v[X' \mapsto (v(X').1,\top)]) \text{ for all } \tr' \in \Theta(X') \text{, and}\\
    &\traceassign[\trv \mapsto \tr_o, v(X').1 \mapsto \tr'],i \models_{\Pi(\ekripke)} \mathit{enc}(\trv,\psi,v[X' \mapsto (v(X').1,\bot)]) \text{ for all } \tr' \notin \Theta(X') \enspace ,
\end{align*}
Since these traces $t_o$ are exactly the traces that satisfy $\traceassign[\trv \mapsto \tr,\rho \mapsto \tr_o],i \models_{\Pi(\ekripke)} \big(\LTLpastglobally(\bigwedge_{p \in \Omega(a)} p_\rho \leftrightarrow p_\trv)$, we can establish the claim for this case by first renaming $\trv$ to $\rho$ in the satisfaction relation inferred via induction hypothesis, and then applying the HyperQPTL semantics and the definition of $\mathit{enc}$.

\smallskip
\noindent\textbf{Case} $\varphi = \exists X''. \psi$: From the semantics of \yltlso we know there is some $Y \subseteq (2^\AP)^\omega$ such that $\ekripke,\tr,i,\Theta[X'' \mapsto Y] \models \psi$. This lets us apply the induction hypothesis to infer for all $X' \in \mathfrak{D}(\Theta[X'' \mapsto Y])$:
\begin{align*}
    &\traceassign[\trv \mapsto \tr, v(X').1 \mapsto \tr'],i \models_{\Pi(\ekripke)} \mathit{enc}(\trv,\psi,v[X' \mapsto (v(X').1,\top)]) \text{ for all } \tr' \in \Theta[X'' \mapsto Y](X') \text{, and}\\
    &\traceassign[\trv \mapsto \tr, v(X').1 \mapsto \tr'],i \models_{\Pi(\ekripke)} \mathit{enc}(\trv,\psi,v[X' \mapsto (v(X').1,\bot)]) \text{ for all } \tr' \notin \Theta[X'' \mapsto Y](X') \enspace .
\end{align*}
Now, let $X'\in \mathfrak{D}(\Theta)$ be any second-order variable from the original assignment $\Theta$. Let $\tr' \in \Theta(X')$. Then, we can infer for all $\tr'' \in T$ that either 
\begin{align*}
    &\traceassign[\trv \mapsto \tr, v(X').1 \mapsto \tr',\rho \mapsto \tr''],i \models_{\Pi(\ekripke)} \mathit{enc}(\trv,\psi,v[X' \mapsto (v(X').1,\top), X'' \mapsto (\rho,\top]) \text{, or}\\
    &\traceassign[\trv \mapsto \tr, v(X').1 \mapsto \tr',\rho \mapsto \tr''],i \models_{\Pi(\ekripke)} \mathit{enc}(\trv,\psi,v[X' \mapsto (v(X').1,\top), X'' \mapsto (\rho,\bot]) \enspace .
\end{align*}
This is because $\traceassign[v(X').1 \mapsto \tr']$ and $v[\mapsto (v(X').1,\top)]$ with $\tr' \in \Theta(X')$ allows us to apply the fact inferred via induction hypothesis for $X''$, and any $\tr'' \in T$ trivially satisfies either $\tr'' \in \Theta[X'' \mapsto Y](X'')$ or $\tr'' \notin \Theta[X'' \mapsto Y](X'')$. By the semantics of HyperQPTL and the definition of $\mathit{enc}$, this lets us infer
\begin{align*}
    &\traceassign[\trv \mapsto \tr, v(X').1 \mapsto \tr'],i \models_{\Pi(\ekripke)} \mathit{enc}(\trv,\exists X'' . \ \psi,v[X' \mapsto (v(X').1,\top)]) \text{ for all } \tr' \in \Theta(X') \enspace .
\end{align*}
We can apply exactly the same steps for any $\tr' \notin \Theta(X')$ to infer the remaining part:
\begin{align*}
    &\traceassign[\trv \mapsto \tr, v(X').1 \mapsto \tr'],i \models_{\Pi(\ekripke)} \mathit{enc}(\trv,\exists X'' . \ \psi,v[X' \mapsto (v(X').1,\bot)]) \text{ for all } \tr' \notin \Theta(X') \enspace .
\end{align*}
\end{proof}
\end{document}